\documentclass{llncs}

\usepackage[dvipsnames]{xcolor}
\usepackage{microtype}
\usepackage{amsmath}
\usepackage{amssymb}
\usepackage{mathpartir}
\usepackage{stmaryrd}
\usepackage{cite}
\usepackage{paralist}

\newcommand{\mathem}{\sf}

\newcommand{\REC}{\mbox{\mathem rec}}

\newcommand{\CASE}{\mbox{\mathem case}}

\newcommand{\OF}{\mbox{\mathem of}}

\newcommand{\bi}{\begin{array}[t]{@{}l@{}}}
\newcommand{\ei}{\end{array}}
\newcommand{\ba}{\begin{array}}
\newcommand{\ea}{\end{array}}
\newcommand{\bda}{\[\ba}
\newcommand{\eda}{\ea\]}
\newcommand{\bp}{\begin{quote}\tt\begin{tabbing}}
\newcommand{\ep}{\end{tabbing}\end{quote}}

\def\ruleform#1{{\setlength{\fboxrule}{1pt}\fbox{\normalsize $#1$}}}

\newcommand{\myirule}[2]{{\renewcommand{\arraystretch}{1.2}\ba{c} #1
                      \\ \hline #2 \ea}}

\newcommand{\rlabel}[1]{\mbox{(#1)}}
\newcommand{\ulabel}[1]{\mbox{(#1)$_u$}}
\newcommand{\dlabel}[1]{\mbox{(#1)$_d$}}
\newcommand{\turns}{\, \vdash \,}
\newcommand{\turnsreg}{\, \vdash_r \,}

\newcommand{\sem}[2]{[\![#1]\!]#2}
\newcommand{\matchRel}[3]{ \, \vdash_m \, #1 : #2 \leadsto #3}
\newcommand{\match}[2]{\textit{match}(#1,#2)}
\newcommand{\Val}{{\cal V}}
\newcommand{\Wrong}{\mbox{\bf W}}
\newcommand{\Kons}{{\cal K}}

\newcommand{\conc}{\cdot}

\newcommand{\deriv}[2]{d_{#2}(#1)} 

\newcommand{\nullable}[1]{{\cal N}(#1)}

\newcommand{\false}{\mathit{False}}
\newcommand{\true}{\mathit{True}}

\newcommand{\ms}[1]{} 
\newcommand{\pt}[1]{} 

\newcommand{\CFE}{Context-Free Expressions}
\newcommand{\cfe}{context-free expression}

\newcommand{\reduce}{\Rightarrow}
\newcommand{\reduceNoFold}[1]{\stackrel{\not\mu #1}{\Rightarrow}}

\newcommand{\flatten}[1]{\mathit{flatten}(#1)}
\newcommand{\unfold}[2]{#2^{#1}}

\newcommand{\reach}[2]{\mathit{reach}(#1,#2)}
\newcommand{\reachRel}[3]{#1 \stackrel{#2}{\leadsto} #3}

\newcommand{\ReachF}[3]{{\cal R}(#1,#2,#3)}
\newcommand{\ReachFUnfold}[2]{{\cal R}'(#1,#2)}
\newcommand{\SubTerm}[1]{{\cal T}(#1)}
\newcommand{\Subst}[1]{{\cal S}(#1)}
\newcommand{\apply}[2]{#1(#2)}
\newcommand{\fapp}[2]{#1 \ (#2)}
\newcommand{\restrict}[2]{#1_{\backslash #2}}
\newcommand{\simp}[1]{\mathit{cnf}(#1)}
\newcommand{\Desc}[1]{D(#1)}
\newcommand{\id}{\mathit{id}}

\newcommand\E{\mathcal E}
\newcommand{\FixP}[2]{{\cal X}_{#1,#2}}
\newcommand{\FixStep}[3]{{X^{#1}}_{#2,#3}}
\newcommand{\EqDom}[2]{{\cal E}_{#1,#2}}
\newcommand{\ReachFunc}[2]{{\cal F}_{#1,#2}}
\newcommand{\CFExp}{\textit{CFE}}
\newcommand{\RExp}{\textit{RE}}

\newcommand{\upTy}[2]{\mbox{\tt U}(#1,#2)} 
\newcommand{\dnTy}[2]{\mbox{\tt D}(#1,#2)} 

\newcommand{\plus}[1]{+#1}
\newcommand{\upCoerce}[4]{#1 \turns^{\!\!\Uparrow} \ #2 : \upTy{#3}{#4}}
\newcommand{\downCoerce}[4]{#1 \turns_{\!\!\Downarrow} \ #2 : \dnTy{#3}{#4}}
\newcommand{\mkEmpty}[1]{\mathit{mkE}(#1)}

\newcommand{\SeqOp}{\textsc{Seq}}
\newcommand\Seq[2]{\textsc{Seq}\ #1\ #2}
\newcommand{\Left}{\textsc{Inl}}
\newcommand{\Right}{\textsc{Inr}}
\newcommand{\Fold}{\textsc{Fold}}

\newcommand{\Eps}{\textsc{Eps}}
\newcommand{\Sym}{\textsc{Sym}}
\newcommand{\Just}{\textit{Just}}
\newcommand{\Nothing}{\textit{Nothing}}
\newcommand{\Maybe}{\textit{Maybe}}

\newcommand\semle\leqq 
\newcommand\semeq\equiv              
\newcommand\syneq=         

\newcommand\DirectDescendant\sqsubset
\newcommand\Descendant\preceq
\newcommand\TrueDescendant\prec

\newcommand\Power{\wp}

\newcommand\BeforeOrEqual\le

\title{A Computational Interpretation of \CFE}

\author{Martin Sulzmann\inst1  \and Peter Thiemann\inst2}
 \institute{
   Faculty of Computer Science and Business Information Systems \\
   Karlsruhe University of Applied Sciences \\
   Moltkestrasse 30, 76133 Karlsruhe, Germany\\
   \email{martin.sulzmann@hs-karlsruhe.de}
   \and 
   Faculty of Engineering, University of Freiburg\\ Georges-K{\"o}hler-Allee
   079, 79110 Freiburg, Germany \\
   \email{thiemann@acm.org}
 }

\begin{document}

\maketitle

\begin{abstract}
  We phrase parsing with context-free expressions as a type inhabitation problem where
  values are parse trees and types are context-free expressions.
  We first show how containment among context-free and regular expressions can be
  reduced to a reachability problem by using a canonical representation of states.
  The proofs-as-programs principle yields a
  computational interpretation of the reachability problem in terms of 
  a coercion that transforms the parse tree for a context-free
  expression into a parse tree for a regular expression. It also
  yields a partial coercion from regular parse trees to context-free
  ones.
  The partial coercion from the trivial language of all words
  to a context-free expression
  corresponds to a predictive parser for the expression.
\end{abstract}

\section{Introduction}
\label{sec:introduction}

In the context of regular expressions, there have been a number of works
which give a \emph{computational} interpretation of regular expressions.
For example, Frisch and Cardelli~\cite{cduce-icalp04} show how to phrase
the regular expression parsing problem as a type inhabitation problem.
Parsing usually means that for an input string that matches a regular expression
we obtain a parse tree which gives a precise explanation which parts of the regular expression
have been matched. By interpreting parse trees as values and regular expressions as types,
parsing can be rephrased as type inhabitation as shown by Frisch and Cardelli.
Henglein and Nielsen~\cite{Henglein:2011:REC:1926385.1926429} as well
Lu and Sulzmann~\cite{DBLP:conf/aplas/LuS04,Sulzmann:2006:TEX:1706640.1706940},
formulate containment of regular expressions as a type
conversion problem. From a containment proof, they derive a
transformation (a type coercion) from parse trees of one regular expression into parse
trees of the other regular expression.

This paper extends these ideas to the setting of context-free expressions.
Context-free expressions extend regular expressions with a least fixed point operator,
so they are effectively equivalent to context-free grammars.
An essential new idea is to phrase the containment problem among \cfe s and regular expressions as
a reachability problem~\cite{Reps:1997:PAV:271338.271343}, where
states are represented by regular expressions and reachable states are
Brzozowski-style derivatives~\cite{321249}. 
By characterizing the reachability problem in terms of a natural-deduction style
proof system, we can apply the proofs-are-programs principle to extract
the coercions that implement the desired transformation between parse trees.

In summary, our contributions are:
\begin{itemize}
\item  an interpretation of \cfe s as types which are inhabited by valid
  parse trees
  (Section~\ref{sec:parsing});

\item a reduction of containment among \cfe s and regular expressions
  to a reachability problem
  (Section~\ref{sec:containment-reachability});

\item a formal derivation of coercions between context-free and
  regular parse trees extracted from  a natural-deduction style proof
  of context-free reachability (Section~\ref{sec:coercions}).
\end{itemize}
%

\emph{The optional appendix contains further details such as proofs etc.}

\section{Preliminaries}
\label{sec:mu-expression}

This section introduces some basic notations including the languages of regular and
context-free expressions and restates some known results for
Brzozowski style derivatives.

Let $\Sigma$ be a finite set of symbols with   $x$, $y$, and $z$
ranging over $\Sigma$.
We write $\Sigma^*$ for the set of finite words over $\Sigma$, 
$\varepsilon$ for the empty word, and
$v \conc w$ for the concatenation of words $v$ and $w$. A language is
a subset of $\Sigma^*$.

\begin{definition}[Regular Expressions]
  The set $\RExp$ of \emph{regular expressions} is defined inductively by
\bda{lcl}
 r,s & ::= & \phi \mid \varepsilon \mid x \in \Sigma \mid (r + s) \mid (r
 \conc s) \mid (r^*)
 \eda
\end{definition}
We omit parentheses by assuming that $^*$ binds tighter than $\conc$ and $\conc$ binds
tighter than $+$. 

\begin{definition}[Regular Languages]
 The meaning function $L$ maps a regular expression to a language.
It is defined inductively as follows:
 \\
 $L(\phi) = \{ \}$.
 $L(\varepsilon) = \{ \varepsilon \}$.
 $L(x) = \{ x \}$.
 $L(r + s) = L(r) \cup L(s)$.
 $L(r \conc s) = \{ v \conc w \mid v \in L(r) \wedge w \in L(s) \}$.
 $L(r^*) = \{ w_1 \conc \ldots \conc w_n \mid n \ge0 \wedge \forall i \in \{1,\ldots,n\}.~ w_i \in L(r) \}$.
\end{definition}
We say that regular expressions $r$ and $s$ are equivalent, $r\semeq s$, if $L(r) = L(s)$.
\begin{definition}[Nullability]
A regular expression $r$ is \emph{nullable} if $\varepsilon \in L(r)$.
\end{definition}

The \emph{derivative} of a regular expression $r$ with respect to some symbol $x$,
written $\deriv{r}{x}$, is a regular expression for the left quotient of $L(r)$ with respect to $x$.
That is, $L(\deriv{r}{x}) = \{ w \in\Sigma^* \mid x \conc w \in L(r) \}$.
A derivative $\deriv{r}{x}$ can be computed by recursion
over the structure of the regular expression $r$.

\begin{definition}[Brzozowski Derivatives~\cite{321249}]
\bda{ll}
\deriv{\phi}{x} = \phi
 & 
 \deriv{\varepsilon}{x} = \phi
 \\
 \\
  \deriv{y}{x} = \left \{ \ba{ll} \varepsilon & \mbox{if $x=y$}
                          \\      \phi        & \mbox{otherwise}
                          \ea
                          \right.
 & \deriv{r + s}{x} = \deriv{r}{x} + \deriv{s}{x}
\\                          
\deriv{r \conc s}{x} = \left \{ \ba{ll}  \deriv{r}{x} \conc s & \mbox{if $\varepsilon \not\in L(r)$}
\\  \deriv{r}{x} \conc s + \deriv{s}{x} & \mbox{otherwise}
\ea
\right.
 & 
\deriv{r^*}{x} = \deriv{r}{x} \conc r^*
\eda
\end{definition}

\begin{example}
  The derivative of $(x+y)^*$ with respect to symbol $x$ is $(\varepsilon + \phi) \conc (x+y)^*$.
  The calculation steps are as follows:
  $$
  \deriv{(x+y)^*}{x}
  = \deriv{x+y}{x} \conc (x+y)^*
  = (\deriv{x}{x} + \deriv{y}{x}) \conc (x+y)^*
  = (\varepsilon + \phi) \conc (x+y)^*
  $$

\end{example}

\begin{theorem}[Expansion \cite{321249}]
  \label{th:representation}
  Every regular expression $r$ can be represented
  as the sum of its derivatives with respect to all symbols. If $\Sigma =
  \{x_1, \dots, x_n\}$, then
  \bda{c}
  r \semeq x_1 \conc \deriv{r}{x_1} + \ldots + x_n \conc \deriv{r}{x_n}
  \ \mbox{($+ \varepsilon$ if $r$ nullable)}
  \eda 
\end{theorem}

\begin{definition}[Descendants and Similarity]
\label{def:desc-sim}
A \emph{descendant} of $r$ is either $r$ itself
or the derivative of a descendant.
We say $r$ and $s$ are \emph{similar}, written $r \sim s$,
if one can be transformed into the other by
finitely many applications of the rewrite rules (Idempotency) $r+r \sim r$,
(Commutativity) $r + s \sim s + r$,
(Associativity) $r + (s + t) \sim (r + s) + t$,
(Elim1) $\varepsilon \conc r \sim r$,
(Elim2) $\phi \conc r \sim \phi$,
(Elim3) $\phi + r \sim r$, and
(Elim4) $r + \phi \sim r$.
\end{definition}
\begin{lemma}
  Similarity is an equivalence relation that respects regular
  expression equivalence: $r \sim s$ implies $r \semeq s$.
\end{lemma}
\begin{theorem}[Finiteness \cite{321249}]
The elements of the set of descendants of a regular expression belong
to finitely many similarity equivalence classes.
\end{theorem}
Similarity rules (Idempotency), (Commutativity), and (Associativity)
suffice to achieve finiteness. Elimination rules are added to
obtain a compact \emph{canonical representative} for
equivalence class of similar regular expressions.
The canonical form is obtained by systematic application
of the similarity rules in Definition~\ref{def:desc-sim}.
We enforce right-associativity of concatenated expressions,
sort alternative expressions according to their size and their first
symbol, and concatenations lexicographically,
assuming an arbitrary total order on $\Sigma$.
We further remove duplicates
and apply elimination rules exhaustively (the details are standard \cite{Grabmayer:2005:UPC:2156157.2156171}). 

\begin{definition}[Canonical Representatives]
\label{def:cnf}  
  For a regular expression $r$, we write $\simp{r}$ to denote
  the canonical representative among all expressions
  similar to $r$.
  We write $\Desc{r}$ for the set of canonical representatives
  of the finitely many dissimilar descendants of $r$.
\end{definition}  

\begin{example}
  We find that $\simp{(\varepsilon + \phi) \conc (x+y)^*} = (x+y)^*$ where we assume $x<y$.
\end{example}

Context-free expressions~\cite{DBLP:journals/corr/WinterBR13}
extend regular expressions with a least fixed point operator $\mu$. 
Our definition elides the Kleene star operator because it can be defined 
with the fixed point operator: $e^* = \mu\alpha. e\conc\alpha +
\varepsilon$. 

\begin{definition}[\CFE]
  Let $A$ be a denumerable set of \emph{placeholders} disjoint from
  $\Sigma$. The  set  $\CFExp$ of 
  \emph{\cfe s} is defined inductively by
  \bda{lcl}
  e, f & ::= & \phi \mid \varepsilon \mid x \in \Sigma \mid \alpha \in A
     \mid e + f \mid e \conc f \mid \mu \alpha. e
  \eda
\end{definition}
We only consider closed \cfe s where
(A) all placeholders are bound by some enclosing $\mu$-operator and (B)
the placeholder introduced by a $\mu$-operator must be distinct from
all enclosing bindings $\mu\alpha$.
Requirement (A) guarantees that reduction of a {\cfe} does not get stuck
whereas requirement (B) ensures that there are no name clashes
when manipulating a \cfe.

While Winter et al \cite{DBLP:journals/corr/WinterBR13} define the
semantics of a context-free expression by coalgebraic means, we define
its meaning with a natural-deduction style big-step semantics.

\begin{definition}[Big-Step Semantics] The reduction relation ${\reduce}
  \subseteq \CFExp \times \Sigma^*$ is defined inductively by the
  following inference rules.
\label{def:red-sem}
  \bda{c}
  \varepsilon \reduce \varepsilon
  \ \ \ \
  x \reduce x
  \ \ \ \
  \myirule{e \reduce w}
          {e + f \reduce w}
  \ \ \ \  
  \myirule{f \reduce w}
          {e + f \reduce w}
  \ \ \ \
  \myirule{e \reduce v \ \ f \reduce w}
          { e \conc f \reduce v \conc w}
  \ \ \ \ 
  \myirule{\apply{[\alpha \mapsto \mu \alpha.e]}{e} \reduce w}
          {\mu \alpha.e \reduce w}
  \eda
In the last rule, we write $\apply{[\alpha \mapsto \mu \alpha.e]}{e}$
to denote the expression obtained by replacing
all occurrences of placeholder $\alpha$ in $e$ by $\mu \alpha.e$. If
$\mu\alpha.e$ is closed, then requirement (B) ensures that there is no
inadvertent capture of placeholders.
  
We further define $L(e) = \{ w \in\Sigma^* \mid e \reduce w \}$.
\end{definition}

As an immediate consequence of the last rule, we see that unfolding does
not affect the language.
\begin{lemma}\label{lemma:unfold-cfe}
  $L (\mu\alpha.e) = L (\apply{[\alpha \mapsto \mu \alpha.e]}{e})$.
\end{lemma}

\begin{definition}[Containment]
  Let $e$ be a context-free expression or regular expression and let $r$ be a regular expression.
  We define $e \leq r$ iff $L(e) \subseteq L(r)$.
\end{definition}

We express partial functions as total functions composed with
lifting as follows. Let $A$ and $B$ be sets.
The set $\Maybe\ B$ consists of elements which are either $\Nothing$ or of the form $\Just\ b$, for
$b\in B$. Thus a total function $f'$ of type $A \to \Maybe\ B$ corresponds uniquely to a partial function
$f$ from $A$ to $B$: for $a\in A$, if $f(a)$ is not defined, then $f' (a) = \Nothing$; if $f(a)=b$ is defined, then
$f' (a) = \Just\ b$; and vice versa.

\section{Parsing as Type Inhabitation}
\label{sec:parsing}

Parsing for regular expressions has been phrased  as a type inhabitation problem~\cite{cduce-icalp04}.
We follow suit and generalize this approach to parsing for context-free expressions.
For our purposes, parse trees are generated by the following grammar.

\begin{definition}[Parse Trees for \cfe s]
\label{def:parse-trees}  
  \bda{lcl}
   p,q & ::= & \Eps \mid \Sym\ x \mid \Left\ p \mid \Right\ q \mid
   \Seq pq \mid \Fold\ p
  \eda
\end{definition}

Like a derivation tree for a context-free grammar, 
a parse tree is a structured representation of the derivation of a word
from some \cfe. 
The actual word can be obtained by flattening the parse tree.

\begin{definition}[Flattening]
  \bda{l@{\qquad}l}
  \flatten{\Eps} = \varepsilon
  &
  \flatten{\Sym\ x} = x
  \\
  \\
  \flatten{\Left\ p} = \flatten{p}
  &
  \flatten{\Right\ q} = \flatten{q}
  \\
  \\
  \flatten{\Seq pq} = \flatten{p} \conc \flatten{q}
  &
  \flatten{\Fold\ p} = \flatten{p}
  \eda
\end{definition}

Compared to derivation trees whose signatures depend on the underlying grammar, parse trees are generic, but their validity depends on the particular \cfe.
The connection between parse trees and \cfe s is made
via the following typing relation where we interpret \cfe s
as types and parse trees as values.
 
 \begin{definition}[Valid Parse Trees, $\turns p : e$]
 \label{def:typing-relation}  

 \bda{ccc}
  \turns \Eps : \varepsilon
  &
  \turns \Sym\ x : x
  &
  \myirule{\turns p : e \ \ \turns q : f}
          {\turns \Seq pq : e \conc f}
  \\
  \\
  \myirule{\turns p : e}
          {\turns \Left\ p: e + f }
  &
  \myirule{\turns p : f}
          {\turns  \Right\ p : e + f}
  &
  \myirule{\turns p : \apply{[\alpha \mapsto \mu \alpha.e]}{e}}
          {\turns \Fold\ p : \mu \alpha.e}
  \eda

\end{definition}

 We consider $\varepsilon$ as a singleton type with
 value $\Eps$ as its only inhabitant.
 The concatenation operator~$\conc$ effectively corresponds to a pair
 where pair values are formed via the binary constructor $\SeqOp$.
 We treat~$+$ as a disjoint sum with the respective
 injection constructors $\Left$ and $\Right$.
 Recursive $\mu$-expressions represent iso-recursive types
 with $\Fold$
 denoting the isomorphism between the unrolling of a recursive type
 and the recursive type itself.

The following results establish that parse trees
obtained via the typing relation can be related to words
derivable in the language of \cfe\ and vice versa.

\begin{lemma}
\label{le:flatten-cfg-proof}
  Let $e$ be a \cfe\ and $w$ be a word. If $e \reduce w$,
  then there exists a parse tree $p$ such that
   $\turns p : e$ where $\flatten{p}=w$.
\end{lemma}

\begin{lemma}
 \label{le:parse-trees-flattening}
  Let $e$ be a \cfe\ and $p$ a parse tree. If $\turns p : e$,
  then $e \reduce \flatten{p}$.
\end{lemma}

\begin{example}
  Let $p = \Fold\ (\Left\ (\Seq (\Sym\ x) \ (\Seq (\Right\  \Eps) \
  (\Sym\ x))))$ be a parse tree and consider
  the expression $e = \mu \alpha. x \conc \alpha + \varepsilon$.
  We find that $\turns p : e$
  and $\flatten{p} = x \conc x$.
\end{example}

Instead of tackling the parsing problem, we solve the more general
problem of coercing parse trees of context-free expressions 
into parse trees of regular expressions and vice versa.

\section{Containment via Reachability}
\label{sec:containment-reachability}

In this section, we consider the problem of determining containment $(e \le r)?$ between a context-free
language represented by some expression $e$ and a regular language represented by regular expression
$r$. This problem is decidable. The standard algorithm constructs a context-free grammar for the intersection $L
(e) \cap \overline{L (r)}$ and tests it for emptiness.

We proceed differently to obtain some computational content from the
proof of containment. We first rephrase the 
containment problem $(e \le r)?$ as a reachability problem. Then, in Section~\ref{sec:coercions}, we
extract computational content by deriving suitable coercions as mappings between
the respective parse trees of $e$ and $r$.

There are coercions in both directions:\vspace{-0.5\baselineskip}
\begin{enumerate}
\item\label{item:1} a total coercion from $L(e)$ to $L(r)$ as a mapping of type $e \to r$ and
\item\label{item:2} a partial coercion from $L(r)$ to $L (e)$ as a mapping of type $r \to \Maybe\ e$,
\end{enumerate}

The partial coercion under~\ref{item:2} can be considered as a
parser specialized to words from $L(r)$. Thus, the partial coercion from $\Sigma^*
\to \Maybe\ e$ is a general parser for $L(e)$.

We say that a regular expression $r'$ is reachable from $e \in \CFExp$ and $r$ if there is some word
$w\in L (e)$ such that $L (r') = w/L(r) = \{ v \in\Sigma^* \mid w\conc v \in L (r)\}$. To obtain a
finite representation, we
define reachability in terms of canonical representatives of derivatives.

\begin{definition}[Reachability]
  Let $e$ be a context-free expression and $r$ a regular expression.
  We define the set of reachable expressions as $\reach{e}{r} = \{
  \simp{\deriv{r}{w}} \mid w\in\Sigma^*,  e \reduce w \}$. 
\end{definition}

\begin{theorem}
\label{th:containment-reachability}  
  Let $e$ be a context-free expression and $r$ be a regular expression.
  Then $e \leq r$ iff each expression in $\reach{e}{r}$
  is nullable.
\end{theorem}

By finiteness of dissimilar descendants the set $\reach{e}{r}$ is finite
and can be computed effectively via a least fixed point construction.
Thus, we obtain a new algorithm for containment by reduction to decidable reachability and nullability.

Instead of showing the least fixed point construction, we give a characterization of the set of
reachable expressions in terms of a natural-deduction style proof system.
The least fixed point construction follows from the proof rules.
\begin{figure}[tp]
  \bda{c}
  \ruleform{\Gamma \turns \reachRel{r}{e}{S}}
  \\
  \\
\rlabel{Eps} \ \Gamma \turns \reachRel{r}{\varepsilon}{\{\simp{r}\}}
\qquad
\rlabel{Phi} \ \Gamma \turns \reachRel{r}{\phi}{\{ \}}
\qquad
\rlabel{Sym} \ \Gamma \turns \reachRel{r}{x}{\{\simp{\deriv{r}{x}}\}}
\\ \\
\rlabel{Alt} \ \myirule{\Gamma \turns \reachRel{r}{e}{S_1}
         \qquad \Gamma \turns \reachRel{r}{f}{S_2}}
        {\Gamma \turns \reachRel{r}{e+f}{S_1 \cup S_2}}
\\ \\
\rlabel{Seq} \ \myirule{\Gamma \turns \reachRel{r}{e}{\{r_1,\ldots,r_n\}}
         \qquad \Gamma \turns \reachRel{r_i}{f}{S_i} \ \mbox{for $i=1,\ldots,n$}
       }
        {\Gamma \turns \reachRel{r}{e\conc f}{S_1 \cup \ldots \cup S_n}}
\\ \\
\rlabel{Rec} \ \myirule{\Gamma \cup \{ \reachRel{r}{\mu\alpha.f}{S} \} \turns \reachRel{r}{\apply{[\alpha\mapsto\mu\alpha.f]}{f}}{S}}
        {\Gamma \turns \reachRel{r}{\mu\alpha.f}{S}}
\ \ \ \
\rlabel{Hyp} \ \myirule{\reachRel{r}{\mu\alpha.f}{S} \in \Gamma}
        {\Gamma \turns \reachRel{r}{\mu\alpha.f}{S}}
\eda
\caption{Reachability proof system}
\label{fig:reachability-proof-system}
\end{figure}

The system in Figure~\ref{fig:reachability-proof-system} defines the judgment  $
\reachRel{r}{e}{S}$ where $e\in\CFExp$, $r$ a regular expression, 
and $S$ is a set of regular expressions in canonical form. It makes use of a set $\Gamma$ of
hypothetical proof judgments of the same form. The meaning of a
judgment is that $S$ (over)approximates $\reach{e}{r}$ (see
upcoming Lemmas~\ref{le:sound-reach-characterization}
and~\ref{le:complete-reach-characterization}). 

The interesting rules are \rlabel{Rec} and \rlabel{Hyp}.
In rule \rlabel{Hyp}, we look up a proof judgment for a context-free expression with topmost operator $\mu$
from the assumption set $\Gamma$.
Such proof judgments are added to $\Gamma$ in rule \rlabel{Rec}.
Hence, we can make use of to be verified proof judgments in subsequent proof steps.
Hence, the above proof system is defined coinductively.
Soundness of the proof system is guaranteed by the fact that we unfold the fixpoint operator $\mu$ in rule \rlabel{Rec}.
We can indeed show soundness and completeness: the set $\reach{e}{r}$ is derivable and any derivable set $S$ is a superset
of $\reach{e}{r}$.

\begin{lemma}
\label{le:sound-reach-characterization}  
  Let $e$ be a context-free expression and $r$ be a regular expression.
  Then, $\turns \reachRel{r}{e}{\reach{e}{r}}$ is derivable.
\end{lemma}

\begin{lemma}
\label{le:complete-reach-characterization}    
  Let $e$ be a context-free expression, $r$ be a regular expression
  and $S$ be a set of expressions such that $\turns \reachRel{r}{e}{S}$.
  Then, we find that $S \supseteq \reach{e}{r}$.
\end{lemma}

\begin{example}
  \label{ex:reach}
  Consider
   $e = \mu\alpha. x \conc (\alpha \conc y) + \varepsilon$  
  and $r = x^* \conc y^*$. It is easy to see that $\reach{e}{r} = \{ r, y^* \}$.
  Indeed, we can verify that $\{\} \turns \reachRel{r}{e}{\{ r, y^* \}}$ is derivable.
   \bda{c}
   \inferrule*[left = \rlabel{Rec}]
   {
     \inferrule*[left = \rlabel{Alt}]
     {
       \inferrule*[left = \rlabel{Seq}]
      {
         \inferrule*[left = \rlabel{Seq}]
         {
           \rlabel{Hyp} \ \  \{ \reachRel{r}{e}{\{r,y^*\}} \} \turns \reachRel{r}{e}{\{r, y^*\}} \checkmark
           \\
           \rlabel{Sym} \ \   \{ \reachRel{r}{e}{\{r, y^*\}} \} \turns \reachRel{r}{y}{\{y^*\}} \checkmark
           \\
           \rlabel{Sym} \ \   \{ \reachRel{r}{e}{\{r, y^*\}} \} \turns \reachRel{y^*}{y}{\{y^*\}} \checkmark           
         }
         {
           \{ \reachRel{r}{e}{\{r,y^*\}} \} \turns \reachRel{r}{e \conc y}{\{ y^*\}}
         }
         \\
         \rlabel{Sym} \ \ \{ \reachRel{r}{e}{\{r, y^*\}} \} \turns \reachRel{r}{x}{\{r\}} \checkmark             
       }
       {
         \{ \reachRel{r}{e}{\{r, y^*\}} \} \turns \reachRel{r}{x\conc (e \conc y)}{\{ y^*\}}
       }
       \\
       \rlabel{Eps} \ \ \{ \reachRel{r}{e}{\{r, y^*\}} \} \turns \reachRel{r}{\varepsilon}{\{ r \}} \checkmark
     }
     {
       \{ \reachRel{r}{e}{\{r, y^*\}} \} \turns \reachRel{r}{x\conc (e \conc y) + \varepsilon}{\{r, y^*\}}
     }         
   }
   {
     \{ \} \turns \reachRel{r}{e}{\{r, y^*\}}                 
   }             
   \eda
   We first apply rule \rlabel{Rec} followed by \rlabel{Alt}.
   One of the premises of \rlabel{Alt} can be verified immediately
   via \rlabel{Eps} as indicated by $\checkmark$.
   For space reasons, we write premises on top of each other.
   Next, we apply \rlabel{Seq} where one of the premises can be verified
   immediately again.
   Finally, we find another application of \rlabel{Seq}.
   $\{ \reachRel{r}{e}{\{r,y^*\}} \} \turns \reachRel{r}{e}{\{r, y^*\}}$
   holds due to \rlabel{Hyp}.
   Because the reachable set contains two elements, $r$ and $y^*$,
   we find two applications of \rlabel{Sym} and we are done.
\end{example}

\begin{example}
  As a special case, consider $e = \mu\alpha.\alpha$
  where $\reach{e}{r} = \{ \}$ for any regular expression $r$.
  The reachability proof system over-approximates and indeed we find that
  $\turns \reachRel{r}{\mu\alpha.\alpha}{S}$ for any $S$
  as shown by the following derivation
  \bda{c}
  \inferrule*[left = \rlabel{Rec}]
             {
                 \rlabel{Hyp} \ \ \{ \reachRel{r}{\mu\alpha.\alpha}{S} \} \turns \reachRel{r}{\mu\alpha.\alpha}{S}
             }
             { \turns \reachRel{r}{\mu\alpha.\alpha}{S}
             }
             
  \eda
\end{example}

\section{Coercions}
\label{sec:coercions}

Our proof system  for the reachability judgment $\reachRel{r}{e}{S}$
in Figure~\ref{fig:reachability-proof-system} provides a coinductive
characterization of the set of reachable expressions. 
Now we apply the proofs-are-programs principle to derive coercions
from derivation trees for reachability. As the proof system is
coinductive, we obtain recursive coercions from applications of the
rules \rlabel{Rec} and \rlabel{Hyp}.

Our first step is to define a term language for coercions, which are
functions on parse trees. This
language turns out to be a lambda calculus (lambda abstraction, function application, variables) with recursion
and pattern matching on parse trees. 

\begin{definition}[Coercion Terms]
  Coercion terms $c$ and
  patterns $pat$ are inductively defined by 
  \bda{lrl}
  c & ::= & v \mid k \mid \lambda v.c \mid c \ c \mid \REC\ x.c
  \mid \CASE\ c \ \OF\ [pat_1 \Rightarrow c_1, \ldots, pat_n \Rightarrow c_n]
  \\
  pat & ::= & v \mid k \ pat_1 \ ... pat_{arity(k)}
  \eda
  where $v$ range overs a denumerable set of variables disjoint from $\Sigma$
  and constructors $k$ are taken from the set
  $\Kons = \{ \Eps, \SeqOp, \Left, \Right, \Fold, \Just, \Nothing, (\_,\_) \}$.
  Constructors $\Eps, \dots, \Fold$ are employed in the formation
  of parse trees. Constructors $\Just$ and $\Nothing$ belong
  to the $\Maybe$ type that arises in the construction of partial coercions. The binary constructor
  $(\_,\_)$  builds a pair.
  The function~$arity(k)$ defines the arity of constructor $k$.
  Patterns are linear (i.e., all pattern variables are distinct) and
  we write $\lambda pat.c$ as a shorthand for
  $\lambda v. \CASE\ v \ \OF\ [pat \Rightarrow c]$.
\end{definition}
We give meaning to coercions in terms of a standard denotational semantics
where values are elements of a complete partial order formed
over the set of parse trees and function space.
We write $\eta$ to denote the mapping from variables to values
and $\sem{c}{\eta}$ to denote the meaning of coercions
where $\eta$ defines the meaning of free variables in $c$.
In case $c$ is closed, we simply write $\sem{c}{}$.

Earlier work shows how to construct coercions that demonstrate containment among regular
expressions~\cite{DBLP:conf/aplas/LuS04,Sulzmann:2006:TEX:1706640.1706940}. These works use a
specialized representation for Kleene star which would require to extend
Definitions~\ref{def:parse-trees} and~\ref{def:typing-relation}. We avoid any special treatment of
the Kleene star by considering $r^*$ an abbreviation for $\mu\alpha.r \conc \alpha + \varepsilon$.
The representations suggested here is isomorphic to the one used in previous 
work~\cite{DBLP:conf/aplas/LuS04,Sulzmann:2006:TEX:1706640.1706940}. We summarize
their main results.  We adopt the convention that $t$ refers to parse trees of regular expressions,
$b$ refers to coercions manipulating regular parse trees.  We write $b : r \rightarrow s$ to denote
a coercion of type $r \rightarrow s$, and we use $\turnsreg t:r$ for the regular typing judgment.

\begin{definition}[Parse Trees for Regular Expressions]
\label{def:regular-parse-trees}  
  \bda{lcl}
  t & ::= & \Eps \mid \Sym\ x \mid \Left\ t \mid \Right\ t \mid
  \Seq tt \mid \Fold\ t
  \eda
\end{definition}

\begin{definition}[Valid Regular Parse Trees, $\turnsreg t : r$]
 \label{def:typing-relation-regex}  
 \begin{mathpar}
  \turnsreg \Eps : \varepsilon

  \turnsreg \Sym\ x : x

  \myirule{\turnsreg t_1 : r \qquad \turnsreg t_2 : s}
  {\turnsreg \Seq {t_1}{t_2} : r \conc s}

  \myirule{\turnsreg t : r}
  {\turnsreg \Left\ t: r + s }
  
  \myirule{\turnsreg t : s}
  {\turnsreg  \Right\ t : r + s}

  \turnsreg \Fold\ (\Right\ \Eps) : r^*

  \myirule{\turnsreg t_1 : r \quad \turnsreg t_2 : r^* }
  {\turnsreg \Fold\ (\Left\ (\Seq {t_1}{t_2})) : r^*}
 \end{mathpar}
\end{definition}

\begin{lemma}[Regular Coercions~\cite{DBLP:conf/aplas/LuS04,Sulzmann:2006:TEX:1706640.1706940}]
  \label{le:regular-coercions}
  Let $r$ and $s$ be regular expressions such that $r \leq s$.
  There is an algorithm to obtain coercions $b_1 : r \rightarrow s$
  and $b_2 : s \rightarrow \Maybe\ r$ such that
  (1) for any $\turnsreg t : r$ we have that $\turnsreg \fapp{b_1}{t} : s$,
  $\sem{\fapp{b_1}{t}}{} = t'$ for some $t'$
  and $\flatten{t} = \flatten{t'}$, and
  (2) for any $\turnsreg t : s$ where $\flatten{t} \in L(r)$ we have
  that $\sem{\fapp{b_2}{t}}{} = \Just\ t'$ for some $t'$ where $\turnsreg t' : r$
  and $\flatten{t} = \flatten{t'}$, and
  (3) for any $\turnsreg t : s$ where $\flatten{t} \not\in L(r)$,
  $\fapp{b_2}{t} =\Nothing$.
\end{lemma}

We refer to $b_1$ as the \emph{upcast} coercion and to $b_2$
as the \emph{downcast} coercion, indicated by  $r \leq^{b_1} s$ and  $r \leq_{b_2} s$, respectively.
Upcasting means that any parse tree for the smaller language can be coerced into a parse tree
for the larger language. On the other hand, a parse tree can only be downcast
if the underlying word belongs to the smaller language.

\begin{figure}[tp]
  \bda{c}
  \ruleform{\upCoerce{\Delta}{c}{e}{r}}
  \\
  \\
\ulabel{Eps} \   \myirule{\simp{r} \leq^{b} r
    \qquad
    c = \lambda (\Eps, t). \fapp{b}{t}}
          {\upCoerce{\Delta}{c}{\varepsilon}{r}}
  \\
  \\
\ulabel{Sym} \  \myirule{ x \conc \simp{\deriv{r}{x}} \leq^{b} r
  \qquad c = \lambda (v, t). \fapp{b}{\Seq {v}{t}}}
          {\upCoerce{\Delta}{c}{x}{r}}
  \\
  \\
\ulabel{Alt} \  \myirule{ \upCoerce{\Delta}{c_1}{e}{r}
    \qquad \upCoerce{\Delta}{c_2}{f}{r}
    \\
    \plus{\reach{e}{r}} \leq_{b_1} \plus{\reach{e+f}{r}}
    \qquad \plus{\reach{f}{r}} \leq_{b_2} \plus{\reach{e+f}{r}}
    \\ c = \ba[t]{l}
          \lambda (p,t). \ba[t]{l}
          \CASE\ p \ \OF\ [
          \\ \ \ \Left\ p_1 \Rightarrow \CASE\ (\fapp{b_1}{t}) \ \OF\ [\Just\ t_1 \Rightarrow c_1 \ (p_1, t_1)],
          \\ \ \ \Right\ p_2 \Rightarrow \CASE\ (\fapp{b_2}{t}) \ \OF\ [\Just\ t_2 \Rightarrow c_2 \ (p_2, t_2)]]
            \ea\ea
          }
          {\upCoerce{\Delta}{c}{e + f}{r}}
  \\
  \\
\ulabel{Seq} \  \myirule{ \upCoerce{\Delta}{c_1}{e}{r}
            \qquad \upCoerce{\Delta}{c_2}{f}{\plus{\reach{e}{r}}}
           \\ c = \lambda (\Seq{p_1}{p_2},t). c_1 \ (p_1, c_2 \ (p_2,t))
          }
          {\upCoerce{\Delta}{c}{e \conc f}{r}}          
  \\
  \\
\ulabel{Rec} \  \myirule{v_{\alpha.e,r} \not\in \Delta
    \qquad
    \upCoerce{\Delta \cup \{ v_{\alpha.e,r} : \upTy{\mu\alpha.e}{ r} \}}{c'}{\apply{[\alpha\mapsto \mu\alpha.e]}{e}}{r}
    \\ c = \REC\ v_{\alpha.e,r}. \lambda (\Fold\ p, t). \fapp{c'}{p,t}
          }
          {\upCoerce{\Delta}{c}{\mu\alpha.e}{r}}
  \\
  \\
\ulabel{Hyp} \  \myirule{(v_{\alpha.e,r} : \upTy{\mu\alpha.e}{r}) \in \Delta}
          {\upCoerce{\Delta}{v_{\alpha.e,r}}{\mu\alpha.e}{r}}
  \eda

  \caption{Reachability upcast coercions}
  \label{fig:upcast-coercions}  
\end{figure}

We wish to extend these results to the containment
$e \leq r$ where $e$ is a context-free expression and $r$ is a regular expression.
In the first step, we build a (reachability upcast) coercion $c$
which takes as inputs a parse tree of $e$ and a proof that
$e$ is contained in $r$. The latter comes in the form of the reachability set  $\reach{e}{r}$, which
we canonicalize to $\plus{\reach{e}{r}}$ as follows:
For a set $R = \{ r_1,\ldots,r_n\}$ of canonical regular expressions,
we define $\plus{R} = \simp{r_1 + \ldots + r_n}$
where we set $\plus{\{\}} = \phi$.

Reachability coercions are derived via the judgment $\upCoerce{\Delta}{c}{e}{r}$, which states that
under environment $\Delta$ an upcast coercion $c$ of type $\upTy{e}{r}$ can be constructed.
Environments $\Delta$ are defined by
$\Delta \ ::= \ \{ \} \mid \{ v : \upTy{e}{r} \} \mid \Delta \cup \Delta$
and record  coercion assumptions, which are needed to construct recursive coercions.  We interpret
$\upTy{e}{r}$ as the type $({e} \times {\plus{\reach{e}{r}}}) \rightarrow r$.
Figure~\ref{fig:upcast-coercions} contains the proof rules which are derived from
Figure~\ref{fig:reachability-proof-system} 
by decorating each rule with an appropriate coercion term.
If $\Delta$ is empty, we  write $\upCoerce{}{c}{e}{r}$ for short.

The proof rules in Figure~\ref{fig:upcast-coercions} are decidable
in the sense that it is decidable if $\upCoerce{\Delta}{c}{e}{r}$ can be derived.
This property holds because proof rules are syntax-directed and $\reach{e}{r}$ is decidable.
We can also attempt to infer $c$ where we either fail or succeed in a finite number of derivation steps.

\begin{lemma}[Upcast Soundness]
  \label{le:le-soundness}
  Let $e$ be a context-free expression and $r$ be a regular expression
  such that $\upCoerce{}{c}{e}{r}$ for some coercion $c$.
  Let $p$ and $t$ be parse trees such that $\turns p : e$
  and $\turnsreg t : \plus{\reach{e}{r}}$ where
  $\flatten{t} \in L(\deriv{r}{\flatten{p}})$.
  Then, we find that $\sem{\fapp{c}{(p,t)}}{} = t'$ for some $t'$
  where $\turns t' : r$ and $\flatten{p} = \flatten{t'}$.
\end{lemma}
The assumption $\flatten{t} \in L(\deriv{r}{\flatten{p}})$ guarantees
that $e$'s parse tree $p$ in combination with $\plus\reach{e}{r}$'s parse tree $t$
allows us to build a parse tree for $r$.

For example, consider rule \ulabel{Alt}.  Suppose $e+f$ parses some input word $w$ because $e$ parses
the word $w$. That is, $w$'s parse tree has the form $p = \Left\ p_1$. As we have proofs that $e\le
r$ and $f \le r$, the downcast $b_1 \ (t)$ cannot fail and yields $\Just \ t_1$. Formally,  we have
$\turns p_1 : e$ and conclude that $\flatten{p}=\flatten{p_1} \in L(e)$.
  By Lemma~\ref{le:parse-trees-flattening}, $e \reduce \flatten{p_1}$ and
  therefore we find that $\deriv{r}{\flatten{p_1}}$ is similar to an element of $\reach{e}{r}$.  
  Because $\flatten{t} \in L(\deriv{r}{\flatten{p}})$ we conclude that $\flatten{t} \in L(\plus{\reach{e}{r}})$.
  By Lemma~\ref{le:regular-coercions}, it must be that
  $\fapp{b_1}{t}  = \Just\ t_1$ for some $t_1$ where $\turns t_1 : \plus{\reach{e}{r}}$. By induction the result holds for $c_1$ and hence we can establish the result
  for $c$.

  In rule \ulabel{Seq}, we exploit the fact
  that $\plus{\reach{e\conc f}{r}} = \plus{\reach{f}{\plus{\reach{e}{r}}}}$.
  So, we use coercion $c_2$ to build a parse tree of $\plus{\reach{e}{r}}$
  given parse trees of $f$ and $\plus{\reach{e\conc f}{r}}$.
  Then, we build a parse tree of $r$ by applying $c_1$
  to parse trees of $e$ and $\plus{\reach{e}{r}}$.
  
  Due to the coinductive nature of the coercion proof system,
  coercion terms may be recursive as evidenced by rule \ulabel{Rec}.
  Soundness is guaranteed by the assumption that the set of reachable
  states is non-empty. As we find a parse tree of that type, progress is made when building the
  coercion for the unfolded $\mu$-expression. Unfolding must terminate because there are only
  finitely many combinations of unfolded subterms of the form $\mu\alpha.e$ and regular expressions
  $r$. The latter are drawn from the finitely many dissimilar descendant of some $r$.
  Hence, resulting coercions must be well-defined as stated
  in the above result. 
  
  \begin{example}\label{ex:upcast}
    We show how to derive $\upCoerce{}{c_0}{e}{r}$ where
   $e = \mu\alpha. x \conc (\alpha \conc y) + \varepsilon$,  
    $r = x^* \conc y^*$ and
     $\reach{e}{r} = \{ r, y^* \}$.
    The shape of the derivation tree corresponds to the derivation
    we have seen in Example~\ref{ex:reach}.
   \bda{c}
   \inferrule*[left = \ulabel{Rec}]
   {
     \inferrule*[left = \ulabel{Alt}]
     {
       \inferrule*[left = \ulabel{Seq}]
      {
         \inferrule*[left = \ulabel{Seq}]
        {
           \ulabel{Hyp} \ \  \upCoerce{\Delta}{c_7}{e}{r} \checkmark
           \\
           \ulabel{Sym} \ \   \upCoerce{\Delta}{c_6}{y}{r+y^*} \checkmark
           \\
           \mbox{}
         }
         {
           \upCoerce{\Delta}{c_5}{e \conc y}{r}
         }
         \\
         \ulabel{Sym} \ \ \upCoerce{\Delta}{c_4}{x}{r} \checkmark             
       }
       {
         \upCoerce{\Delta}{c_3}{x\conc (e \conc y)}{r}
       }
       \ulabel{Eps} \ \ \upCoerce{\Delta}{c_2}{\varepsilon}{r} \checkmark
     }
     {
       \upCoerce{\Delta}{c_1}{x\conc (e \conc y) + \varepsilon}{r}
     }         
   }
   {
     \upCoerce{}{c_0}{e}{r}
   }             
   \eda

   We fill in the details by following the derivation tree from bottom to top.
   We set $\Delta = \{ v_{\alpha.e,r} : \upTy{e}{r} \}$.
   From the first \ulabel{Rec} step we conclude
   $c_0 = \REC\ v_{\alpha.e,r}. \lambda (\Fold\ p, t). \fapp{c_1}{p,t}$.
   Next, we find  \ulabel{Alt} which yields
\bda{lcl}
c_1 & = & \ba[t]{l}
          \lambda (p,t). \ba[t]{l}
          \CASE\ p \ \OF\ [
          \\ \ \ \Left\ p_1 \Rightarrow \CASE\ (\fapp{b_1}{t}) \ \OF\ [\Just\ t_1 \Rightarrow c_3 \ (p_1, t_1)],
          \\ \ \ \Right\ p_2 \Rightarrow \CASE\ (\fapp{b_2}{t}) \ \OF\ [\Just\ t_2 \Rightarrow c_2 \ (p_2, t_2)]]
          \ea\ea
\eda          

We consider the definition of the auxiliary regular (downcast)
coercions $b_1$ and $b_2$. We have that
$\plus{\reach{x\conc (e\conc y) + \varepsilon}{r}} = r + y^*$,
$\plus{\reach{\varepsilon}{r}} = r$ and
$\plus{\reach{x\conc (e\conc y)}{r}} = y^*$.
Hence, we need to derive $y^* \leq_{b_1} r + y^*$ and $r \leq_{b_2} r + y^*$.
   
Recall the requirement (2) for downcast coercions.
See Lemma~\ref{le:regular-coercions}.
We first consider $y^* \leq_{b_1} r + y^*$. The right component of the sum
can be straightforwardly coerced into a parse tree of $y^*$.
For the left component we need to check that the leading part
is effectively empty. Recall that Kleene star is represented
in terms of $\mu$-expressions. Following
Definition~\ref{def:typing-relation-regex}, an empty parse tree
for Kleene star equals $\Fold\ (\Right\ \Eps)$.
Thus, we arrive at the following definition for $b_1$.
\bda{c}

\myirule{b_1 = \ba[t]{l}
                \lambda t. \ba[t]{l}
                \CASE\ t \ \OF\ [
                  \\ \ \ \Left\ (\Seq{ (\Fold\ \Right\ \Eps)} v) \Rightarrow \Just\ v,
                  \\ \ \ \Left\ v \Rightarrow \Nothing,
                  \\ \ \ \Right\ v \Rightarrow \Just\ v]
                            \ea
               \ea
        }
        {y^* \leq_{b_1} r + y^*}
\eda

The derivation of $r \leq_{b_2} r + y^*$ follows a similar pattern.
As both expressions $r$ and $r + y^*$ are equal, the downcast never fails here.
\bda{c}
\myirule{b_2 = \ba[t]{l}
                \lambda t. \ba[t]{l}
                \CASE\ t \ \OF\ [
                  \\ \ \ \Left\ v \Rightarrow \Just\ v,
                  \\ \ \ \Right\ v \Rightarrow \Just\ (\Seq {(\Fold\ \Right\ \Eps)} v)]
                            \ea
               \ea}
        {r \leq_{b_2} r + y^*}       
\eda

Next, consider the premises of the \ulabel{Alt}
rule. For  $\upCoerce{\Delta}{c_2}{\varepsilon}{r}$
by definition $c_2 = \lambda (\Eps,t). \fapp{b_3}{t}$
where $r \leq^{b_3} r$ which can be satisfied by $b_3 = \lambda v.v$.
For $\upCoerce{\Delta}{c_3}{x\conc (e \conc y)}{r}$
we find by definition $c_3 = \lambda (\Seq{p_1}{p_2},t). c_4 \ (p_1, c_5 \ (p_2,t))$.

It follows some \ulabel{Seq} step where we first
consider $\upCoerce{\Delta}{c_4}{x}{r}$.
By definition of \ulabel{Sym} and $\simp{\deriv{x}{r}} = r$ we have that
$c_4 = \lambda (v, t). \fapp{b_4}{\Seq {v}{t}}$
where $x \conc r \leq^{b_4} r$. Recall $r = x^* \conc y^*$.
So, upcast $b_4$ injects $x$ into $x^*$'s parse tree.
 Recall the representation
of parse trees for Kleene star in Definition~\ref{def:typing-relation-regex}.
\bda{c}
b_4 = \lambda (\Seq v {(\Seq {t_1}{ t_2})}.\Seq{ (\Fold\ (\Left\ (\Seq v { t_1})))} { t_2}
\eda

Next, we consider $\upCoerce{\Delta}{c_5}{e \conc y}{r}$
where we find another \ulabel{Seq} step.
Hence, $c_5 = \lambda (\Seq{p_1}{p_2},t). c_7 \ (p_1, c_6 \ (p_2,t))$.
By \ulabel{Hyp}, we have that $c_7 = v_{\alpha.e,r}$.
To obtain $\upCoerce{\Delta}{c_6}{y}{r+y^*}$ we apply
another \ulabel{Sym} step and therefore
$c_6 = \lambda (v, t). \fapp{b_5}{\Seq {v}{t}}$.
The regular (upcast) coercion $b_5$ is derived from $y \conc y^* \leq^{b_5} r + y^*$
because $\simp{\deriv{r+y^*}{y}} = y^*$. Its definition is as follows.
\bda{c}
 b_5 = \lambda (\Seq v  t). \Right\ (\Fold\ (\Left\ (\Seq v t)))
\eda
This completes the example.
\end{example}
  
  \begin{remark}[Ambiguities]
    Example~\ref{ex:upcast} shows that coercions may be ambiguous
    in the sense that there are several choices for the resulting parse trees.
    For example, in the construction of
    the regular (upcast) coercion $y \conc y^* \leq^{b_5} x^*\conc y^* + y^*$
    we choose to inject $y$ into the right component of the sum.
    The alternative is to inject $y$ into the left component
    by making the $x^*$ part empty.    
    \bda{c}
     b_5' = \lambda (\Seq v  t). \Left\ (\Seq {(\Fold\ (\Right\ \Eps))} { (\Fold\ (\Left\ (\Seq v t)))})
     \eda
     Both are valid choices. To obtain deterministic behavior of coercions
     we can apply a disambiguation strategy (e.g., favoring
     left-most alternatives). 
     A detailed investigation of this topic is beyond the scope
     of the present work.
  \end{remark}

  Based on Lemma \ref{le:le-soundness} we easily obtain
  an upcast coercion to transform $e$'s parse tree into
  a parse tree of $r$. As $e \le r$ if
all elements in $\reach{e}{r}$ are nullable,
we simply need to provide
an empty parse tree
for $\plus{\reach{e}{r}}$.
The upcoming definition of $\mkEmpty{}$ supplies such parse trees. It requires to check
for nullability of context-free expression.
This check is decidable as shown by the following definition.

\begin{definition}[CFE Nullability]
  \begin{align*}
    \nullable{\phi} =\nullable{x} & = \false  & \nullable{e+f} & =  \nullable{e} \vee \nullable{f} \\
    \nullable{\varepsilon} & = \true             & \nullable{e\conc f} & =  \nullable{e} \wedge \nullable{f} \\
    \nullable{\alpha} &= \false                    &    \nullable{\mu\alpha.e} & =  \nullable{e}
  \end{align*}
\end{definition}

\begin{lemma}
  \label{le:cfe-nullability}
  Let $e$ be a context-free expression.
  Then, we have that $\nullable{e}$ holds iff $\varepsilon \in L(e)$.
\end{lemma}

Based on the nullability check, we can derive empty parse trees (if they exist).

\begin{definition}[Empty Parse Tree]
  \bda{llcll}
  \mkEmpty{\varepsilon}& = \Eps
   &&
      \mkEmpty{e + f} & = \left \{ \ba{ll}   \Left\ \mkEmpty{e} & \mbox{if $\nullable{e}$}
                                    \\       \Right\ \mkEmpty{f} & \mbox{otherwise}
                             \ea
                             \right.
  \\
  \\
  \mkEmpty{\mu\alpha.e} & = \Fold\ \mkEmpty{e}
   &&
    \mkEmpty{e \conc f} & = \Seq{\mkEmpty{e}}{ \mkEmpty{f}}
  \eda   
\end{definition}

\begin{lemma}
 \label{le:mkEmpty} 
  Let $e$ be a context-free expression such that $\nullable{e}$.
  Then, we find that $\turns \mkEmpty{e} : e$ and $\flatten{\mkEmpty{e}} = \varepsilon$.
\end{lemma}

We summarize the construction of upcast coercions
for context-free and regular expressions in containment relation.

\begin{theorem}[Upcast Coercions]
\label{th:upcast-coercions}  
  Let $e$ be a context-free expression and $r$ be a regular expression such that
  $e \leq r$ and $\upCoerce{}{c'}{e}{r}$ for some coercion $c'$.
  Let $c = \lambda x. c' \ (x, \mkEmpty{\plus{\reach{e}{r}}})$.
  Then, we find that $c$ is well-typed with type $e \rightarrow r$ where
  for any $\turns p : e$ we have that $\sem{\fapp{c}{p}}{} = t'$
  for some $t'$ where
  and $\turns t' : r$ and $\flatten{p} = \flatten{t'}$.
\end{theorem}

\begin{figure}[tp]
  \bda{c}
  \ruleform{\downCoerce{\Delta}{c}{e}{r}}
  \\
  \\
  \dlabel{Eps}
  \
  \myirule{r \leq^{b} \simp{r}
    \qquad \ c = \lambda t. \Just\ (\Eps, \fapp{b}{t})}
          {\downCoerce{\Delta}{c}{\varepsilon}{r}}
  \\
  \\
  \dlabel{Sym} \
  \myirule{x \conc \simp{\deriv{r}{x}} \leq_b r
    \\
    c = \lambda t. 
    \CASE\ (\fapp{b}{t})\ \OF\ 
    [\Nothing \Rightarrow \Nothing,\ 
    \Just\ (\Seq {x'} {t'}) \Rightarrow \Just\ (x', t')]
  }
          {\downCoerce{\Delta}{c}{x}{r}}
  \\
  \\
  \dlabel{Alt} \
  \myirule{\downCoerce{\Delta}{c_1}{e}{r} \qquad \downCoerce{\Delta}{c_2}{f}{r}
             \\ \plus{\reach{e}{r}} \leq^{b_1} \plus{\reach{e+f}{r}}
             \qquad  \plus{\reach{f}{r}} \leq^{b_2} \plus{\reach{e+f}{r}}
             \\ c = \lambda t. \ba[t]{l}
                           \CASE\ (\fapp{c_1}{t}) \ \OF
                      \\   {}[\Nothing \Rightarrow \CASE\ (\fapp{c_2}{t}) \ \OF
                      \\   \ \ \ \ \ \ \ \ \ \ \ \ [\Nothing \Rightarrow \Nothing,
                      \\   \ \ \ \ \ \ \ \ \ \ \ \ \Just\ (p_2, t_2) \Rightarrow \Just\ (\Right\ p_2, \fapp{b_2}{t_2})],
                      \\   \Just\ (p_1, t_1) \Rightarrow \Just\ (\Left\ p_1, \fapp{b_1}{t_1})]
                    \ea
          }
          {\downCoerce{\Delta}{c}{e+f}{r}}
  \\
  \\
  \dlabel{Seq} \
  \myirule{\downCoerce{\Delta}{c_1}{e}{r} \qquad \downCoerce{\Delta}{c_2}{f}{\plus{\reach{e}{r}}}
    \\ c = \lambda t. \ba[t]{l}
                  \CASE\ (\fapp{c_1}{t}) \ \OF
             \\   {}[\Nothing \Rightarrow \Nothing,
             \\   \Just\ (p_1, t_1) \Rightarrow \CASE\ (\fapp{c_2}{t_1}) \ \OF
             \\   \ \ \ \ \ \ \ \ \ \ \ \ \ [\Nothing \Rightarrow \Nothing,
             \\   \ \ \ \ \ \ \ \ \ \ \ \ \ \Just\ (p_2, t_2) \Rightarrow \Just\ (\Seq{p_1}{p_2},t_2)]]
           \ea
          }
          {\downCoerce{\Delta}{c}{e \conc f}{r}}
  \\
  \\
  \dlabel{Rec} \
  \myirule{ v_{\alpha.e,r}  \not\in \Delta
    \ \ \
    \downCoerce{\Delta\cup \{  (v_{\alpha.e,r} : \dnTy{\mu\alpha.e}{r}) \}}{c'}{\apply{[\alpha\mapsto \mu\alpha.e]}{e}}{r}
    \\ c =  \REC\ v_{\alpha.e,r}. \lambda t.    \ba[t]{l}
    \CASE\ (\fapp{c'}{t}) \ \OF
    \\  {}[\Nothing \Rightarrow \Nothing,
    \\  \Just\ (p', t') \Rightarrow \Just\ (\Fold\ p', t')]      
    \ea
            }
          {\downCoerce{\Delta}{c}{\mu\alpha.e}{r}}
  \\
  \\
   \dlabel{Hyp} \
    \myirule{(v_{\alpha.e,r} : \dnTy{\mu\alpha.e}{r}) \in \Delta}
          {\downCoerce{\Delta}{v_{\alpha.e,r}}{\mu\alpha.e}{r}}
  \eda
  \caption{Reachability downcast coercions}
  \label{fig:downcast-coercions}  
\end{figure}  

In analogy to the construction of upcast coercions,
we can build a proof system for the construction of downcast coercions.
Each such downcast coercion $c$ has type $\dnTy{e}{r}$
where $\dnTy{e}{r}$ corresponds to
$r \rightarrow \Maybe\ (e \times \plus{\reach{e}{r}})$.
That is, a parse tree of $r$ can possibly be coerced
into a parse tree of $e$ and some residue
which is a parse tree of $\plus{\reach{e}{r}}$.
See Figure~\ref{fig:downcast-coercions}.

Rule \dlabel{Eps} performs a change in representation.
The downcast will always succeed.
Rule \dlabel{Sym} applies the regular downcast $b$ to split
$r$'s parse tree into $x$ and the parse tree of the (canonical) derivative.
The resulting downcast will not succeed if there is no leading $x$.

In case of a sum, rule \dlabel{Alt} first tests if we can downcast
$r$'s parse tree into a parse tree of the left component $e$
and $\plus{\reach{e}{r}}$. If yes, we upcast
$\plus{\reach{e}{r}}$'s parse tree into
a parse tree of $\plus{\reach{e+f}{r}}$.
Otherwise, we check if a downcast into $f$
and $\plus{\reach{f}{r}}$ is possible.

In rule \dlabel{Seq}, we first check if we can obtain
parse trees for $e$ and residue $\plus{\reach{e}{r}}$.
Otherwise, we immediately reach failure.
From $\plus{\reach{e}{r}}$'s parse tree we then attempt
to extract $f$'s parse tree and residue $\plus{\reach{f}{\plus{\reach{e}{r}}}}$
which we know is equivalent to $\plus{\reach{e\conc f}{r}}$.
Hence, we combine the parse trees of $e$ and $f$ via $\Seq$
and only need to pass through the residue.

As in case of upcast coercions, downcast coercions may be recursive.
See rules \dlabel{Rec} and \dlabel{Hyp}.
In case the downcast yields the parse tree $p'$ of the unfolding, we apply $\Fold$.
The residue $t'$ can be passed through as we find that
$\plus{\reach{\mu \alpha.e}{r}} = \plus{\reach{\apply{[\alpha\mapsto \mu\alpha.e]}{e}}{r}}$.

\begin{lemma}[Downcast Soundness]
  Let $e$ be a context-free expression and $r$ be a regular expression
  such that 
  $\downCoerce{\Delta}{c}{e}{r}$ for some coercion $c$.
  Let $t$ be such that
  $\turns t : r$ and $\sem{\fapp{c}{t}}{} = \Just\ (p,t')$
  for some $p$ and $t'$.
  Then, we have that $\turns p : e$, $\turns t' : \plus{\reach{e}{r}}$
  and $\flatten{t} = \flatten{p}$.
\end{lemma}

\begin{example}
\label{ex:downcast}  
 We consider the derivation of $\downCoerce{}{c_0}{e}{r}$ where
   $e = \mu\alpha. x \conc (\alpha \conc y) + \varepsilon$,  
    $r = x^* \conc y^*$ and
 $\reach{e}{r} = \{ r, y^* \}$.
 The downcast coercion attempts to turn a parse of $r$
 into a parse tree of $e$ and some residual parse tree of $\plus\reach{e}{r}$.
 The construction is similar to Example~\ref{ex:upcast}.
 We consider the downcast coercions resulting from \dlabel{Rec} and \dlabel{Alt}.
  \bda{l}
  c_0 : \dnTy{e}{r}
  \ = \  \REC\ v_{\alpha.e,r}. \lambda t.    \ba[t]{l}
    \CASE\ (\fapp{c_1}{t}) \ \OF
    \\  {}[\Nothing \Rightarrow \Nothing,
    \\  \Just\ (p', t') \Rightarrow \Just\ (\Fold\ p', t')]      
    \ea
  \\
  \\    
  c_1 : \dnTy{x \conc (e \conc y) + \varepsilon}{r}
   \ = \ \lambda t. \ba[t]{l}
                           \CASE\ (\fapp{c_2}{t}) \ \OF
                      \\   {}[\Nothing \Rightarrow \CASE\ (\fapp{c_3}{t}) \ \OF
                      \\   \ \ \ \ \ \ \ \ \ \ \ \ [\Nothing \Rightarrow \Nothing,
                      \\   \ \ \ \ \ \ \ \ \ \ \ \ \Just\ (p_2, t_2) \Rightarrow \Just\ (\Right\ p_2, \fapp{b_2}{t_2})],
                      \\   \Just\ (p_1, t_1) \Rightarrow \Just\ (\Left\ p_1, \fapp{b_1}{t_1})]
                      \ea
\\
\mbox{where} \ r \leq^{b_1} r+y^* \ \ r \leq^{b_2} r+y^*
  \ \  b_1 = \Right \ \ b_2 = \Left
\eda

The auxiliary coercion $c_2$ greedily checks for a leading symbol $x$. Otherwise, we pick the base case \dlabel{Eps}
where the entire input becomes the residue.
This is dealt with by coercion $c_3$.
\bda{l}                        
  c_3 : \dnTy{\varepsilon}{r}
    \ = \  \lambda t. \Just\ (\Eps, \fapp{b_3}{t})
    \\ \mbox{where} \ \simp{r} = r \ \ b_3 =\lambda x.x \ \ r \leq^{b_3} \simp{r})

\eda

Coercion $c_2$ first checks for $x$, then recursively calls (in essence) $c_0$,
followed by a check for $y$.
Here are the details.
\bda{l}                        
  c_2 : \dnTy{x \conc (e \conc y)}{r}
   \ = \ \lambda t. \ba[t]{l}
                  \CASE\ (\fapp{c_4}{t}) \ \OF
             \\   {}[\Nothing \Rightarrow \Nothing,
             \\   \Just\ (p_1, t_1) \Rightarrow \CASE\ (\fapp{c_5}{t_1}) \ \OF
             \\   \ \ \ \ \ \ \ \ \ \ \ \ \ [\Nothing \Rightarrow \Nothing,
             \\   \ \ \ \ \ \ \ \ \ \ \ \ \ \Just\ (p_2, t_2) \Rightarrow \Just\ (\Seq{p_1}{p_2},t_2)]]
             \ea
\eda             
Auxiliary coercion $c_4$ checks for $x$ and any residue
is passed on to coercion $c_5$.
\bda{l}
  c_4 : \dnTy{x}{r}
  \ = \ \lambda t. 
    \CASE\ (\fapp{b_4}{t})\ \OF\ 
    [\Nothing \Rightarrow \Nothing,\ 
    \Just\ (\Seq {x'} {t'}) \Rightarrow \Just\ (x', t')]
  \\
  \mbox{where} \ \simp{\deriv{r}{x}} = r  \ \ \ x \conc r \leq_{b_4} r
  \\ \\ \ \ \ \ \ \ \ \ b_4 = \ba[t]{l}
            \lambda\Seq{t_1}{t_2}.
            \ba[t]{l}
            \CASE\ t_1 \ \OF\ [
              \\ \ \ \Fold\ \Right\ \Eps \Rightarrow \Nothing,
              \\ \ \ \Fold\ \Left\ \ (\Seq{t_3}{t_4}) \Rightarrow
                   \Just\ (\Seq{t_3}{(\Seq{t_4}{t_3})})]
            \ea
          \ea
\eda

In coercion $c_5$, we check for $e$ which then leads to the recursive call.
\bda{l}
  c_5 : \dnTy{e \conc y}{r}
  \ = \ \lambda t. \ba[t]{l}
                  \CASE\ (\fapp{c_7}{t}) \ \OF
             \\   {}[\Nothing \Rightarrow \Nothing,
             \\   \Just\ (p_1, t_1) \Rightarrow \CASE\ (\fapp{c_6}{t_1}) \ \OF
             \\   \ \ \ \ \ \ \ \ \ \ \ \ \ [\Nothing \Rightarrow \Nothing,
             \\   \ \ \ \ \ \ \ \ \ \ \ \ \ \Just\ (p_2, t_2) \Rightarrow \Just\ (\Seq{p_1}{p_2},t_2)]]
             \ea
   \\
  c_7 : \dnTy{e}{r} \ = \ v_{\alpha.e,r}
\eda

Finally, coercion $c_6$ checks for $y$

\bda{l}
  c_6 : \dnTy{y}{r + y^*}
  \ = \ \lambda t. \ba[t]{l}
    \CASE\ (\fapp{b_6}{t})\ \OF\ [ 
    \\ \ \ \Nothing \Rightarrow \Nothing,\ 
    \\ \ \ \Just\ (\Seq {x'} {t'}) \Rightarrow \Just\ (x', t')]
    \ea
    \\ \mbox{where} \ \simp{\deriv{r+y^*}{y}} = y^* \ y \conc y^* \leq_{b_6} r + y^*
    \\ \ \ \ \ \ \ \ \ \ b_6 = \lambda t.
             \ba[t]{l}
             \CASE\ t \ \OF\ [
               \\ \ \ \Left \ \Seq{t_1}{t_2} \Rightarrow \fapp{b_6'}{t_2},
               \\ \ \ \Right\ t \Rightarrow \fapp{b_6'}{t} ]            
             \ea
  \\ \ \ \ \ \ \ \ \ \ y \conc y^* \leq_{b_6'} y^*  
  \\ \ \ \ \ \ \ \ \ \ b_6' = \lambda t.
            \ba[t]{l}
            \CASE\ t \ \OF\ [
              \\ \ \ \Fold\ \Right\ \Eps \Rightarrow \Nothing,
              \\ \ \ \Fold\ \Left\ \ (\Seq{t_1}{t_2}) \Rightarrow
                   \Just\ (\Seq{t_1}{t_2}) ]
            \ea
  \eda

Consider input $t = \Seq{t_1}{t_2}$ where
$t_1 = \Fold\ (\Left\ \Seq{x}{(\Fold\ (\Right\ \Eps))})$,
$t_2 = \Fold\ (\Left\ \Seq{y}{(\Fold\ (\Right\ \Eps))})$,
$\turns t : r$ and $\flatten{t} = x \conc y$.
Then $\sem{\fapp{c_0}{t}} = \Just\ (p, t')$ where
$p= \Fold\ (\Left\ \Seq{x}{(\Seq{(\Fold\ (\Right\ \Eps))}{y})})$
and $\turns p : e$ and $\flatten{p} = x \conc y$.
For residue $t'$ we find $\flatten{t'} = \varepsilon$.
This completes the example.
\end{example}

Any context-free expression $e$ is contained in the regular language $\Sigma^*$.
We wish to derive a downcast coercion for this containment
which effectively represents a parser for $L(e)$.
That is, the parser maps a parse tree for $w \in \Sigma^*$
(which is isomorphic to $w$) to a parse tree $\turns p: e$ with $w = \flatten{p}$ if $w\in L(e)$.
However, our parser, like any predictive parser, is sensitive to the shape of context-free expressions.
So, we need to syntactically restrict the class of context-free expressions on which our parser can be applied.

\begin{definition}[Guarded Context-Free Expressions]
\label{def:guarded-cfe}  
  A context-free expression is \emph{guarded} if the expression
  is of the following shape:
  \bda{lcl}
  e, f & ::= & \phi \mid \varepsilon \mid x \in \Sigma \mid \alpha \in A
  \mid e + f \mid e \conc f \mid \mu \alpha. g
  \\
  g & ::= & x \conc e \mid \varepsilon \mid x \conc e + g
  \eda
  where for each symbol $x$ there exists at most one guard $x \conc e$ in $g$.
\end{definition}

For any context-free expression we find an equivalent guarded variant.
This follows from the fact that
guarded expressions effectively correspond to context-free grammars
in Greibach Normal Form.
We additionally impose the conditions that guards $x$ are unique
and $\varepsilon$ appears last.
This ensures that the parser leaves no residue behind.

\begin{theorem}[Predictive Guarded Parser]
  Let $e$ be a guarded context-free expression and $r$ be a regular expression
  such that $e \leq r$ and
  $\downCoerce{\Delta}{c'}{e}{r}$ for some coercion $c'$.
  Let $c = \lambda x. \CASE\ (c' \ (x)) \ \OF \
  [\Nothing \Rightarrow \Nothing, \Just\ (p,t') \Rightarrow \Just\ p]$
  Then, we find that $c$ is well-typed with type
  $r \rightarrow \Maybe\ e$ and terminates for any input $t$ 
   $\turns t : r$
  If $\sem{\fapp{c}{t}}{} = \Just\ p$
  for some $p$, then we have that $\turns p : e$
  and $\flatten{t} = \flatten{p}$.
\end{theorem}

Guardedness is essential as shown by the following examples.
Consider $e' = \mu\alpha. \varepsilon + x \conc (\alpha \conc y)$,  
$r = x^* \conc y^*$.
The difference to $e$ from Example~\ref{ex:downcast} is that
subexpression $\varepsilon$ appears in leading position.
Hence, the guardedness condition is violated.
The downcast coercion for this example (after some simplifications) has the form
$ c_0 = \lambda t. (\Fold\ (\Left\ \Eps), t).$
%
As we can see no input is consumed at all. We return the trivial parse term 
and the residue $t$ contains the unconsumed input.
As an example for a non-terminating parser consider $e' = \mu \alpha. \alpha \conc x + \varepsilon$ and
$r =(x+y)^*$. Again the guardedness condition is violated
because subexpression $\alpha$ is not guarded.
The coercion resulting from $\downCoerce{}{c_0'}{e'}{r}$
has (after some simplifications) the following form
$ c_0' = \REC\ v. \lambda t. \CASE\ v \ (t) \ \OF\ \dots.$
%
Clearly, this parser is non-terminating which is no surprise as the context-free expression is left-recursive.

\section{Related Work and Conclusion}
\label{sec:related-work}

Our work builds upon prior work
in the setting of regular expressions
by Frisch and Cardelli~\cite{cduce-icalp04},
Henglein and Nielsen~\cite{Henglein:2011:REC:1926385.1926429} and
Lu and Sulzmann~\cite{DBLP:conf/aplas/LuS04,Sulzmann:2006:TEX:1706640.1706940}, as well as Brandt
and Henglein's coinductive characterization of recursive type equality and subtyping
\cite{Brandt:1998:CAR:291677.291678}. 
We extend these ideas to the case of context-free expressions and their parse trees.

There are simple standard methods to construct predictive parsers (e.g., recursive descent etc)
contained in any textbook on compiler construction~\cite{Aho:1986:CPT:6448}.
But the standard methods are tied to parse from a single regular input language, $\Sigma^*$,
whereas our approach provides specialized parsers from an arbitrary regular language.
These parsers will generally be more deterministic, fail earlier, etc. because they are exploiting knowledge about the input. 

Based on our results we obtain a predictive parser for guarded context-free expressions.
Earlier works in this area extend Brzozowski-style derivatives~\cite{321249}
to the context-free setting while we use plain regular expression
derivatives in combination with reachability.
See the works by Krishnaswami~\cite{cfe-typed},
Might, Darais and Spiewak~\cite{DBLP:conf/icfp/MightDS11}
and Winter, Bonsangue, and Rutten \cite{DBLP:journals/corr/WinterBR13}.
Krishnaswami~\cite{cfe-typed} shows how to elaborate general context-free expressions
into some equivalent guarded form and how to transform guarded parse trees into their
original representation.
We could integrate this elaboration/transformation step into our approach
to obtain a geneneral, predictive parser for context-free expressions.

Marriott, Stuckey, and Sulzmann~\cite{DBLP:conf/aplas/MarriottSS03}
show how containment among context-free languages and regular languages
can be reduced to a reachability problem~\cite{Reps:1997:PAV:271338.271343}.
While they represent languages as context-free grammars and DFAs, we rely on context-free
expressions, regular expressions, and specify
reachable states in terms of Brzozowski-style derivatives~\cite{321249}.
This step is essential to obtain a characterization of the
reachability problem in terms of a natural-deduction style
proof system.
By applying the proofs-are-programs principle we derive
upcast and downcast coercions to transform parse trees
of context-free and regular expressions.
These connections are not explored in any prior work.

\section*{Acknowledgments}

We thank the APLAS'17 reviewers for their constructive feedback.

\bibliography{main}

\begin{thebibliography}{10}

\bibitem{Aho:1986:CPT:6448}
A.~V. Aho, R.~Sethi, and J.~D. Ullman.
\newblock {\em Compilers: Principles, Techniques, and Tools}.
\newblock Addison-Wesley Longman Publishing Co., Inc., Boston, MA, USA, 1986.

\bibitem{Brandt:1998:CAR:291677.291678}
M.~Brandt and F.~Henglein.
\newblock Coinductive axiomatization of recursive type equality and subtyping.
\newblock {\em Fundam. Inf.}, 33(4):309--338, Apr. 1998.

\bibitem{321249}
J.~A. Brzozowski.
\newblock Derivatives of regular expressions.
\newblock {\em J. ACM}, 11(4):481--494, 1964.

\bibitem{cduce-icalp04}
A.~Frisch and L.~Cardelli.
\newblock Greedy regular expression matching.
\newblock In {\em Proc. of ICALP'04}, pages 618-- 629. Spinger-Verlag, 2004.

\bibitem{Grabmayer:2005:UPC:2156157.2156171}
C.~Grabmayer.
\newblock Using proofs by coinduction to find "traditional" proofs.
\newblock In {\em Proc.\ of CALCO'05}, pages 175--193. Springer-Verlag, 2005.

\bibitem{Henglein:2011:REC:1926385.1926429}
F.~Henglein and L.~Nielsen.
\newblock Regular expression containment: Coinductive axiomatization and
  computational interpretation.
\newblock In {\em Proc.\ of POPL'11}, pages 385--398. ACM, 2011.

\bibitem{cfe-typed}
N.~R. Krishnaswami.
\newblock A typed, algebraic approach to parsing,.
\newblock \verb+https://www.cl.cam.ac.uk/~nk480/parsing.pdf+, 2017.

\bibitem{DBLP:conf/aplas/LuS04}
K.~Z.~M. Lu and M.~Sulzmann.
\newblock An implementation of subtyping among regular expression types.
\newblock In {\em Proc.\ of APLAS'04}, volume 3302 of {\em LNCS}, pages 57--73.
  Springer, 2004.

\bibitem{DBLP:conf/aplas/MarriottSS03}
K.~Marriott, P.~J. Stuckey, and M.~Sulzmann.
\newblock Resource usage verification.
\newblock In {\em Proc.\ of APLAS'03}, volume 2895 of {\em LNCS}, pages
  212--229. Springer, 2003.

\bibitem{DBLP:conf/icfp/MightDS11}
M.~Might, D.~Darais, and D.~Spiewak.
\newblock Parsing with derivatives: a functional pearl.
\newblock In {\em Proc.\ of ICFP'11}, pages 189--195. {ACM}, 2011.

\bibitem{Reps:1997:PAV:271338.271343}
T.~Reps.
\newblock Program analysis via graph reachability.
\newblock In {\em Proc.\ of ILPS'97}, pages 5--19, Cambridge, MA, USA, 1997.
  MIT Press.

\bibitem{Sulzmann:2006:TEX:1706640.1706940}
M.~Sulzmann and K.~Z.~M. Lu.
\newblock A type-safe embedding of {XDuce} into {ML}.
\newblock {\em Electron. Notes Theor. Comput. Sci.}, 148(2):239--264, 2006.

\bibitem{DBLP:journals/corr/WinterBR13}
J.~Winter, M.~M. Bonsangue, and J.~J. M.~M. Rutten.
\newblock Coalgebraic characterizations of context-free languages.
\newblock {\em Logical Methods in Computer Science}, 9(3), 2013.

\end{thebibliography}

\clearpage
\appendix
\section*{Appendix}
\section{Notation}

\begin{definition}[Substitution]
  We write $[\alpha_1 \mapsto e_1, \ldots, \alpha_n \mapsto e_n]$
  to denote a substitution mapping variables $\alpha_i$ to expressions $e_i$.
  We maintain the invariant that the
  free variables (if any) in $e_i$ are disjoint from $\alpha_i$.
  That is, the substitutions we consider here are \emph{idempotent}.

  Let $\psi = [\alpha_1 \mapsto e_1, \ldots, \alpha_n \mapsto e_n]$.
  Then, we define $\psi(\alpha) = e$ if there exists $i$ such that
  $\alpha_i = \alpha$ and $e_i = e$.
  This extends to expressions in the natural way.

Let $\psi  = [\alpha_1 \mapsto e_1, \ldots, \alpha_n \mapsto e_n]$.
Then, we define $\restrict{\psi}{\alpha_1} = [\alpha_2 \mapsto e_2, \ldots, \alpha_n \mapsto e_n]$.
  
  We write $\id$ to denote the empty substitution $[]$.
  Let $\psi_1 = [\alpha_1 \mapsto e_1, \ldots, \alpha_n \mapsto e_n]$
  and $\psi_2 = [\beta_1 \mapsto f_1, \ldots, \beta_m \mapsto f_m]$
  be two substitutions where $\alpha_i$ and $\beta_j$ are distinct
  and the free variables in $e_i$ and $f_j$ are disjoint from $\alpha_i$
  and $\beta_j$.
  Then, we define $\psi_1 \sqcup \psi_2 =
  [\alpha_1 \mapsto e_1, \ldots, \alpha_n \mapsto e_n,
    \beta_1 \mapsto f_1, \ldots, \beta_m \mapsto f_m]$.
\end{definition}  

\section{Coercion Semantics}
\label{sec:coercion-semantics}

Values are elements of a complete partial order $\Val$ which is defined
as the least solution of the following domain equation, where ``$+$'' and ``$\Sigma$'' stand for the
lifted sum of of domains. The distinguished element $\Wrong$ (wrong) 
will be used to indicate errors. The resulting domain yields a non-strict interpretation.

\begin{definition}[Values]
\bda{ccc}
    \Val &  = & \{\Wrong\} + (\Val \rightarrow \Val) + \sum_{k \in \Kons} (\{k\} \times \Val_1 \times\dots\times \Val_{arity(k)})
\eda
\end{definition}
We write $a$ to denote values, i.e.~elements in $\Val$.
We abuse notation by writing $\__\bot $ for the injections into the sum on the left-hand side (the
summand is clear from the argument) and $\downarrow $ (drop) 
for their right-inverses.

To define the meaning of coercions, we first establish
a semantic match relation among values and patterns to obtain the binding
of pattern variables. We write $\sqcup$ for the disjoint union of environments $\eta$ which map
variables to values.

\begin{definition}[Pattern Matching, $\matchRel{a}{pat}{\eta}$]
  \bda{c}
    \matchRel{a}{v}{[v \mapsto a]}
    \ \ \
    \myirule{
      \matchRel{\downarrow a_1}{pat_1}{\eta_1} \ \dots\ 
      \matchRel{\downarrow a_n}{pat_n}{\eta_n}
    \quad n = arity(k)}
            {\matchRel{(k, a_1,  \dots, a_n)}{k \ pat_1 \ ... pat_n}{\eta_1 \sqcup ... \sqcup \eta_n}}
  \eda
\end{definition}
Pattern matching may fail, for example, in case of differences
in constructors and number of arguments.
We write $\Just\ \eta = \match{v}{pat}$ to indicate that $\matchRel{v}{pat}{\eta}$
is derivable. Otherwise $\match{v}{pat} = \Nothing$.

\begin{definition}[Coercion Semantics, $\sem{c}{\eta}$]
  \bda{l}
  \sem{v}{\eta} \  = \ \eta(v)
  \\
  \\
  \sem{\lambda v.c}{\eta} \ = \ (\lambda y.\sem{c}{(\eta \sqcup [v \mapsto y])})_\bot
  \\
  \\
  \sem{k \ c_1 ... c_{arity(k)}}{\eta} \ = \ (k, \sem{c_1}{\eta}, \dots,\sem{c_{arity(k)}}{\eta})_\bot
  \\
  \\
  \sem{c_1 \ c_2}{\eta} \ = \ \mbox{if $\downarrow\sem{c_1}{\eta} \in \Val \rightarrow \Val$ then $\downarrow(\sem{c_1}{\eta}) \ (\sem{c_2}{\eta})$ else $\Wrong_\bot$}
  \\
  \\
  \sem{\CASE\ c \ \OF\ [pat_1 \Rightarrow c_1, \ldots, pat_n \Rightarrow c_n]}{\eta} \ = \
  \\ \ \ \ \mbox{let $y = {\downarrow{\sem{c}{\eta}}}$ in }
  \\ \ \ \ \mbox{if $\Just\ \eta_1 = \match{y}{pat_1}$ then $\sem{c_1}{(\eta \sqcup \eta_1)}$}
  \\ \ \ \ ...
  \\ \ \ \ \mbox{else if $\Just\ \eta_n = \match{y}{pat_n}$ then $\sem{c_1}{(\eta \sqcup \eta_n)}$}
  \\ \ \ \ \mbox{else $\Wrong_\bot$}
  \eda

\end{definition}

\section{Downcast Example}
\label{sec:downcast-example}

\subsection{$\downCoerce{}{c_0}{\mu\alpha. x \conc \alpha + \varepsilon}{(x+y)^*}$}

  Consider $e = \mu\alpha. x \conc \alpha + \varepsilon$ and $r = (x+y)^*$.
  We have that $\reach{e}{r} = \{ r \}$. Recall that $\simp{\deriv{(x+y)^*}{x}} = (x+y)^*$.

  The following derivation tree verifies that $\{\} \turns \reachRel{r}{e}{\{r\}}$.
   \bda{c}
   \inferrule*[left = \rlabel{Rec}]
   {
     \inferrule*[left = \rlabel{Alt}]
     {
       \inferrule*[left = \rlabel{Seq}]
       {
         \rlabel{Sym} \ \ \{ \reachRel{r}{e}{\{r\}} \} \turns \reachRel{r}{x}{\{r\}} \checkmark                                 
        \\              
        \rlabel{Hyp} \ \  \{ \reachRel{r}{e}{\{r\}} \} \turns \reachRel{r}{e}{\{r\}} \checkmark
       }
       {
         \{ \reachRel{r}{e}{\{r\}} \} \turns \reachRel{r}{x \conc e}{\{r\}}
       }
       \\
       \rlabel{Eps} \ \ \{ \reachRel{r}{e}{\{r\}} \} \turns \reachRel{r}{\varepsilon}{\{ r \}} \checkmark
     }
     {
       \{ \reachRel{r}{e}{\{r\}} \} \turns \reachRel{r}{x \conc e+ \varepsilon}{\{r\}}
     }         
   }
   {
     \{ \} \turns \reachRel{r}{e}{\{r\}}                 
   }             
   \eda

   The derivation of $\downCoerce{}{c_0}{e}{r}$ follows the shape of the above derivation tree.
   We find the following coercions where for clarity we label them with the corresponding
   rule names. Definitions of auxiliary regular coercions are found at the end.

   \bda{llcl}
   \dlabel{Rec} & c_0 & = &
   \ba[t]{l}
   \REC\ v_{\alpha}. \lambda t. \ba[t]{l}
   \CASE\ c_1 \ (t) \ \OF
   \\ \ \ [\Nothing \Rightarrow \Nothing,
   \\ \ \ [\Just\ (p',t') \Rightarrow \Just\ (\Fold\ p',t')]
   \ea
   \ea
   \\
   \\
   \dlabel{Alt} & c_1 & = &
   \ba[t]{l}
   \lambda t. \ba[t]{l}
   \CASE\ c_2 \ (t) \ \OF
   \\ \ \ [\Nothing \Rightarrow \CASE\ (c_3 \ (t)) \ \OF
   \\ \ \ \ \ \ \ \ \ \ \ \ \ [\Nothing \Rightarrow \Nothing,
   \\ \ \ \ \ \ \ \ \ \ \ \ \ \Just\ (p_2,t_2) \Rightarrow \Just\ (\Right\ p_2, b_2 \ (t_2))],
   \\ \ \ \Just\ (p_1,t_1) \Rightarrow \Just \ (\Left\ p_1, b_1 \ (t_1))]
   \ea
   \ea
   \\
   \\
   \dlabel{Seq} & c_2 & = &
   \ba[t]{l}
   \lambda t. \ba[t]{l}
   \CASE\ (c_3 \ (t)) \ \OF
   \\ \ \ [\Nothing \Rightarrow \Nothing,
   \\ \ \ \Just\ (p_1,t_1) \Rightarrow \CASE\ (c_4 \ (t_1)) \ \OF
   \\ \ \ \ \ \ \ \ \ \ \ \ \ [\Nothing \Rightarrow \Nothing,
   \\ \ \ \ \ \ \ \ \ \ \ \ \ \Just\ (p_2,t_2) \Rightarrow \Just\ (\Seq{p_1}{p_2}, t_2)]]     
   \ea
   \ea
   \\
   \\
   \dlabel{Sym} & c_3 & = &
   \lambda t. 
    \CASE\ (\fapp{b_4}{t})\ \OF\ 
    [\Nothing \Rightarrow \Nothing,\ 
      \Just\ (\Seq {x'} {t'}) \Rightarrow \Just\ (x', t')]
   \\
   \\
   \dlabel{Hyp} & c_4 & = & v_{\alpha}    
   \\
   \\
   \dlabel{Eps} & c_5 & = & \lambda t. \Just\ (\Eps, b_3 \ (t))    
   \eda

      \bda{c}
   \myirule{
     \reach{e \conc x}{r} = r \ \ \reach{e\conc x + \varepsilon}{r} = r
     \ \ b_1 = \lambda t.t
     }{\plus{\reach{e \conc x}{r}} \leq^{b_1} \plus{\reach{e\conc x + \varepsilon}{r}}}
   \eda

   \bda{c}
   \myirule{
     \reach{\varepsilon}{r} = r \ \ \reach{e\conc x + \varepsilon}{r} = r
     \ \ b_2 = \lambda t.t
     }{\plus{\reach{\varepsilon}{r}} \leq^{b_2} \plus{\reach{e\conc x + \varepsilon}{r}}}
   \eda

   \bda{c}
   \myirule{b_3 = \lambda t.t \ \ \simp{r}=r}{r \leq^{b_3} \simp{r}}    
   \eda

   \bda{c}
   cons = \REC\ v. \ba[t]{l}
   \lambda x. \lambda xs. \ba[t]{l}
   \CASE\ xs \ \OF\ [
     \\ \ \ \Fold\ (\Right\ \Eps) \Rightarrow \Fold\ (\Left\ (\Seq{x}{(\Fold\ (\Right\ \Eps))})),
     \\ \ \ \Fold\ (\Left\ (\Seq{y}{ys})) \Rightarrow \Fold\ (\Left\ (\Seq{x}{(\Fold\ (\Left\ (\Seq{y}{ys})))}))]
    \ea
   \ea
   \eda

   \bda{c}
   \myirule{
     b_4 = \REC\ v. \ba[t]{l}
     \lambda t. \ba[t]{l}
                \CASE\ t \ \OF\ [
                  \\ \ \ \Fold\ (\Right\ \Eps) \Rightarrow \Nothing,
                  \\ \ \ \Fold\ (\Left\ (\Seq{(\Right \ (\Sym \ y))}{t_2})) \Rightarrow \Nothing,
                  \\ \ \ \Fold\ (\Left\ (\Seq{(\Left \ (\Sym \ x))}{(\Fold\ (\Right\ \Eps))})) \Rightarrow
                  \\ \ \ \ \ \ \ \ \ \ \ \ \Just\ (\Seq{(\Sym\ x)}{(\Fold\ (\Right\ \Eps))}),
                  \\ \ \ \Fold\ (\Left\ (\Seq{(\Left \ (\Sym \ x))}{t_2})) \Rightarrow
                  \\ \ \ \ \ \ \CASE\ (v \ (t_2)) \ \OF\ [
                  \\ \ \ \ \ \ \ \ \Nothing \Rightarrow \Nothing,
                  \\ \ \ \ \ \ \ \ \Just\ (\Seq{(\Sym\ x_2)}{t_3}) \Rightarrow \Just\ (\Seq{(\Sym\ x)}{(cons \ (\Sym\ x_2) \ t_3)})]]
             \ea
           \ea
     }
     {x \conc r \leq_{b_4} r}

   \eda
   
\subsection{$\downCoerce{}{c_0'}{\mu\alpha. \varepsilon + x \conc \alpha}{(x+y)^*}$}

By reusing the above calculations we obtain

\bda{llcl}
   \dlabel{Rec} & c_0 & = &
   \ba[t]{l}
   \REC\ v_{\alpha}. \lambda t. \ba[t]{l}
   \CASE\ c_1' \ (t) \ \OF
   \\ \ \ [\Nothing \Rightarrow \Nothing,
   \\ \ \ [\Just\ (p',t') \Rightarrow \Just\ (\Fold\ p',t')]
   \ea
   \ea
   \\
   \\
   \dlabel{Alt} & c_1' & = &
   \ba[t]{l}
   \lambda t. \ba[t]{l}
   \CASE\ c_2' \ (t) \ \OF
   \\ \ \ [\Nothing \Rightarrow \CASE\ (c_3' \ (t)) \ \OF
   \\ \ \ \ \ \ \ \ \ \ \ \ \ [\Nothing \Rightarrow \Nothing,
   \\ \ \ \ \ \ \ \ \ \ \ \ \ \Just\ (p_2,t_2) \Rightarrow \Just\ (\Right\ p_2, b_1 \ (t_2))],
   \\ \ \ \Just\ (p_1,t_1) \Rightarrow \Just \ (\Left\ p_1, b_2 \ (t_1))]
   \ea
   \ea
   \\
   \\
   \dlabel{Eps} &   c_2' & = & \lambda t. \Just\ (\Eps, b_3 \ (t))
   \\
   \\
   \dots
\eda

So, by unfolding and removing dead code we find that

\bda{c}
c_0' = \lambda t. (\Eps, t)
\eda

\subsection{$\downCoerce{}{c_0''}{\mu\alpha. \alpha \conc x + \varepsilon}{(x+y)^*}$}

Via similar reasoning we find that

\bda{llcl}
   \dlabel{Rec} & c_0'' & = &
   \ba[t]{l}
   \REC\ v_{\alpha}. \lambda t. \ba[t]{l}
   \CASE\ c_1'' \ (t) \ \OF
   \\ \ \ [\Nothing \Rightarrow \Nothing,
   \\ \ \ [\Just\ (p',t') \Rightarrow \Just\ (\Fold\ p',t')]
   \ea
   \ea
   \\
   \\
   \dlabel{Alt} & c_1'' & = &
   \ba[t]{l}
   \lambda t. \ba[t]{l}
   \CASE\ c_2'' \ (t) \ \OF
   \\ \ \ [\Nothing \Rightarrow \CASE\ (c_3'' \ (t)) \ \OF
   \\ \ \ \ \ \ \ \ \ \ \ \ \ [\Nothing \Rightarrow \Nothing,
   \\ \ \ \ \ \ \ \ \ \ \ \ \ \Just\ (p_2,t_2) \Rightarrow \Just\ (\Right\ p_2, b_1 \ (t_2))],
   \\ \ \ \Just\ (p_1,t_1) \Rightarrow \Just \ (\Left\ p_1, b_2 \ (t_1))]
   \ea
   \ea
   \\
   \\
   \dlabel{Hyp} &   c_2'' & = & v_{\alpha}
   \\
   \\
   \dots
\eda

\section{Least Fixed Point Construction for $\reach{e}{r}$}
\label{sec:compute-fix-point}

To compute $\reach e r$, we need to compute $\reach{e'}{r'}$ for all
subterms of $e$ and for certain $r$. To capture this notion exactly,
we define a function to compute the set of subterms of a context-free
expressions. 

\begin{definition}[Subterms]
  \bda{c}
    \SubTerm{\varepsilon} = \{ \varepsilon \}
    \ \ \ 
    \SubTerm{\phi} = \{ \phi \}
    \ \ \
    \SubTerm{x} = \{ x \}
    \\
    \\
    \SubTerm{e + f} = \{ e + f \} \cup \SubTerm{e} \cup \SubTerm{f}
    \ \ \
    \SubTerm{e \conc f} = \{ e \conc f \} \cup \SubTerm{e} \cup \SubTerm{f}
    \\
    \\
    \SubTerm{\alpha} = \{ \alpha \}
    \ \ \
    \SubTerm{\mu \alpha.e} = \{ \mu\alpha.e \} \cup \SubTerm{e}
 \eda   
  \end{definition}

\begin{lemma}
For any context-free expression $e$, the set $\SubTerm{e}$ is finite.
\end{lemma}

We write $R$ and $S$ to denote sets of regular expressions.
We write $E$ to denote an equation of the form $(e,r) = R$.
We can view a set $\E$ of such equations as a mapping from pairs $(e,r)$ to $R$.
If $(e,r) = R \in \E$ then we write $\E(e,r)$ to denote $R$.
If no such equation exists in $\E$, then we set $\E(e,r) = \emptyset$.

We define $\EqDom{e}{r}$ as the set of equations where the
pairs range over subterms of $e$ and derivatives of $r$ and map to
sets of descendants of  $r$.
\bda{c}
\EqDom{e}{r} = \{ (f,s) = S \mid f \in \SubTerm{e}, s \in \Desc{r}, S \subseteq \Desc{r} \}
\eda
  
For two sets of equations $\E_1,\E_2 \subseteq \EqDom{e}{r}$, we define $\E_1 \leq \E_2$
if for each $(f,s) \in \SubTerm{e} \times \Desc{r}$ we have that $\E_1(f,s) \subseteq \E_2(f,s)$.
This definition makes $(\EqDom{e}{r}, \leq)$ a complete, finite lattice with
least element $\{ \}$ and greatest element
$\{ (f,s) = \Desc{r} \mid (f,s) \in \SubTerm{e} \times \Desc{r} \}$.

\pt{Essentially, we are talking about the powerset lattice of
  $\SubTerm{e} \times \Desc{r} \times \Desc{r}$ considered as a
  partial function. }

Next, we define
the reachability step function $\ReachF{}{}{}$
which operates on $\SubTerm{e} \times \Desc{r} \times \Power\EqDom{e}{r}$
and yields a subset of $\Desc{r}$.
\begin{definition}[Reachability Step]
  \bda{lcl}
  \ReachF{\phi}{r}{\E} & = & \{ \}
  \\
  \ReachF{\varepsilon}{r}{\E} & = & \{ \simp{r} \}
  \\
  \ReachF{x}{r}{\E} & = & \{ \simp{\deriv{r}{x}} \}
  \\
  \ReachF{e + f}{r}{\E} & = & \ReachF{e}{r}{\E} \cup \ReachF{f}{r}{\E}
  \\
  \ReachF{e \conc f}{r}{\E} & = & \bigcup_{s \in \ReachF{e}{r}{\E}} \ReachF{f}{s}{\E}
  \\
  \ReachF{\alpha}{r}{\E} & = & \E(\mu\alpha.e,r)
  \\
  \ReachF{\mu\alpha.e}{r}{\E} & = & \ReachF{e}{r}{\E}
  \eda
  For the second last case, we assume that each variable $\alpha$
  can be linked to its surrounding scope $\mu\alpha.e$.
  This is guaranteed by the fact that we consider a fixed set of subterms.
  
  We sometimes write $\ReachF{f}{\ReachF{e}{r}{\E}}{\E}$
  as a shorthand for $\bigcup_{s \in \ReachF{e}{r}{\E}} \ReachF{f}{s}{\E}$.
\end{definition}

\begin{definition}[Reachability Function]
  Let $e$ be a context-free expression and $r$ be a regular expression.
  We define $\ReachFunc{e}{r} : \EqDom{e}{r} \rightarrow \EqDom{e}{r}$ as follows:
   \begin{align*}
     \ReachFunc{e}{r}(\E) =
      \{ (f,s) = \E (f,s) \cup \ReachF{f}{s}{\E} \mid (f,s) \in
     \SubTerm{e} \times \Desc{r} \}
   \end{align*}
\end{definition}

\begin{lemma}
  Function $\ReachFunc{e}{r}$ is well-defined and monotonic with respect to the ordering $\leq$.
\end{lemma}
\begin{proof}
  All calls to $\ReachFunc{e}{r}$ yield well-defined calls to $\ReachF{}{}{}$
  as elements in $\EqDom{e}{r}$ cover all cases on which $\ReachF{}{}{}$.
  Furthermore,
  the range of function $\ReachF{}{}{}$ is the set $\Desc{r}$.
  Hence, computation will never get stuck.

  For monotonicity, we need to show that
  if $\E_1 \leq \E_2$ then $\ReachFunc{e}{r}(\E_1) \leq
  \ReachFunc{e}{r}(\E_2)$, which holds if $\ReachF{}{}{}$ is monotonic
  in the last parameter. The proof for monotonicity of
  $\ReachF{}{}{}$ is by easy induction over the first parameter.
  \qed
\end{proof}  

The Knaster-Tarski Theorem guarantees that the least fixpoint of $\ReachFunc{e}{r}$
exists.
Let $\FixP{e}{r}$ denote the least fixpoint. That is,
$\FixP{e}{r} = \bigcup_{i=0}^\infty\ReachFunc{e}{r}^{i}(\bot)$
\ms{i still don't get this part, there's a sup missing!}
\pt{the ordering is set inclusion, so the sup $=$ union}
\ms{but 'set inclusion' only arises on the right-hand side, i'm still confused!}
where $\bot = \{ (f,s) = \{ \} \mid (f,s) \in \SubTerm{e} \times \Desc{r} \}$.
Hence, in the sequence of elements
$\FixStep{0}{e}{r} = \bot$ and
$\FixStep{n+1}{e}{r} = \ReachFunc{e}{r}(\FixStep{n}{e}{r})$,
we find
$\FixP{e}{r} = \FixStep{m}{e}{r}$ for some $m \geq 0$
where $\FixStep{m}{e}{r} = \FixStep{m+k}{e}{r}$ for all $k \geq 0$.

For sets $R$ and $S$ of regular expressions, we define $R \sim S$
if for each $r\in R$ we find $s\in S$ where $r \sim s$ and vice versa.

For the proof of $\FixP{e}{r}(e,r) \sim \reach{e}{r}$ to go through,
we need to include subterms in $\SubTerm{e}$.
As these subterms  contain free variables $\alpha$, we need
to map these variables to their corresponding definition.

\begin{definition}[Binding of $\mu$-Expressions]  
  Let $e$ be a context-free expression.
  We build a substitution which maps bound variables $\alpha$ in $e$
  to their corresponding definition.    
  \bda{c}
  \Subst{\varepsilon} = \id
  \ \ \
  \Subst{\phi} = \id
    \ \ \
  \Subst{x} = \id
    \ \ \
  \Subst{\alpha} = \id
    \\
    \\
  \Subst{e + f} = \Subst{e} \sqcup \Subst{f}
    \ \ \
  \Subst{e \conc f} = \Subst{e} \sqcup \Subst{f}
    \\
    \\
    \Subst{\mu \alpha.e} = [\alpha_1 \mapsto \apply{\psi}{e_1},\ldots,\alpha_n \mapsto \apply{\psi}{e_n},\alpha \mapsto \mu\alpha.e]
    \\ \mbox{where $\Subst{e} = [\alpha_1 \mapsto e_1, \ldots, \alpha_n \mapsto e_n] \ \ \ \psi = [\alpha \mapsto \mu\alpha.e]$}
  \eda
  \pt{Actually, the last line is the same as $\Subst{\mu\alpha.e} = [\alpha \mapsto \mu
    \alpha.e] \circ \Subst{e}$}
  \ms{nope, i'm doing something different here}
  \pt{I beg to differ, you are expanding the definition of composition for substitutions}  
\end{definition}

The substitution $\SubTerm{e}$ is well-defined as
variables $\alpha$ introduced by $\mu\alpha$ are distinct by assumption.
See the cases for concatenation and alternation.
In case of $\mu\alpha.e$, we first build $\Subst{e}=[\alpha_1 \mapsto e_1, \ldots, \alpha_n \mapsto e_n]$ where $e_i$ may only refer to $\alpha$ or other variables
but not to $\alpha_i$. Hence, the application of $\psi(e_i)$ to maintain
the invariant.
 \pt{That's true, but the actual argument is more subtle
  than this and requires a variable ordering. I worked out the details
  for another paper where I did the PDA construction from derivatives.}
 \ms{see above}
 \ms{PDA, i think you might have mentioned but i forgot}
 \pt{Yeah, I got tired}

Some helper statements which follow by definition.

\begin{lemma}
  \label{le:psi}
  Let $e$ be a context-free expression and $\mu\alpha.f \in \SubTerm{e}$.
  Let $\psi = \Subst{e}$.
  Then, we have that
  $\apply{[\alpha \mapsto \mu\alpha.\apply{\restrict{\psi}{\alpha}}{f}]}{\apply{\restrict{\psi}{\alpha}}{f}} = \apply{\psi}{f}$.
\end{lemma}

\begin{lemma}
  \label{le:fix-p-mu}
   \label{le:fix-p-alpha}
  Let $e$ be a context-free expression and $r$ be a regular expression.
  Let $(\mu\alpha.f,s) \in \SubTerm{e} \times \Desc{r}$.
  Then, we have that $\FixP{e}{r}(\mu\alpha.f,s) = \FixP{e}{r}(f,s) = \FixP{e}{r}(\alpha,s)$.
\end{lemma}

\begin{lemma}
  Let $e$ be a context-free expression and $r$ be a regular expression.
  Let $(f,s) \in \SubTerm{e} \times \Desc{r}$.
  Then, we have that $\ReachF{f}{s}{\FixP{e}{r}} = \FixP{e}{r}(f,s)$.
\end{lemma}

The generalized statement. 

\begin{lemma}
 \label{le:reach-soundness-left} 
 Let $e$ be a context-free expression and $r$ be a regular expression.
 Let $\psi = \SubTerm{e}$, $w$ be a word,
 $(f,s) \in \SubTerm{e} \times \Desc{r}$ such that $\apply{\psi}{f} \reduce w$.
  Then, there exists $t\in \FixP{e}{r}(f,s)$ such that $t \sim \deriv{s}{w}$.  
\end{lemma}
\begin{proof}
  By induction on the derivation $\apply{\psi}{f} \reduce w$
  and observing the various shapes of $\apply{\psi}{f}$.

  {\bf Case} $\apply{\psi}{\mu\alpha.f}$:

  By definition $\apply{\psi}{\mu\alpha.f} = \mu\alpha.\apply{\restrict{\psi}{\alpha}}{f}$.
  Hence, we find that
  $\myirule{\apply{[\alpha \mapsto \restrict{\psi}{\alpha}]}{\apply{\restrict{\psi}{\alpha}}{f}} \reduce w}
  {\mu\alpha.\apply{\restrict{\psi}{\alpha}}{f} \reduce w}
  $
  By Lemma~\ref{le:psi}, we find that
  $\apply{[\alpha \mapsto \mu\alpha.\apply{\restrict{\psi}{\alpha}}{f}]}{\apply{\restrict{\psi}{\alpha}}{f}} = \apply{\psi}{f}$.

  By induction, there exists $t \in \FixP{e}{r}(f,s)$ where $t \sim \deriv{s}{w}$.
  By Lemma~\ref{le:fix-p-mu}, $\FixP{e}{r}(f,s) = \FixP{e}{r}(\mu\alpha.f,s)$
  and thus we are done.

  {\bf Case} $\apply{\psi}{\alpha}$:

  There must exist $\mu\alpha.f \in \SubTerm{e}$.
  Hence, this case can be reduced to the one above and
  we find that $t \in \FixP{e}{r}(\mu\alpha,s)$ where $t \sim \deriv{s}{w}$.
  By Lemma~\ref{le:fix-p-alpha}, $\FixP{e}{r}(\mu\alpha.f,s) = \FixP{e}{r}(\alpha,s)$ and thus we are done again.

  {\bf Case} $\apply{\psi}{f_1 + f_2}$:

  Suppose $\myirule{\apply{\psi}{f_1} \reduce w}
                   {\apply{\psi}{f_1} + \apply{\psi}{f_2} \reduce w}
          $.
 By induction, there exists $t \in \FixP{e}{r}(f_1,s)$ where $t \sim \deriv{s}{w}$.
  By construction $\FixP{e}{r}(f_1+f_2,s) \supseteq \FixP{e}{r}(f_1,s)$
  (recall the definition of $\ReachFunc{}{}{}$). Thus, we are done.
  Same reasoning applies for the (sub)case $\myirule{\apply{\psi}{f_2} \reduce w}
                   {\apply{\psi}{f_1} + \apply{\psi}{f_2} \reduce w}
          $.
  
  {\bf Case} $\apply{\psi}{f_1 \conc f_2}$:

  Consider $\myirule{\apply{\psi}{f_1} \reduce w_1 \ \apply{\psi}{f_2} \reduce w_2}
  {\apply{\psi}{f_1} \conc \apply{\psi}{f_2} \reduce w_1 \conc w_2}$.
  By induction on the case $(f_1,s)$, there exists $t_1 \in \FixP{e}{r}(f_1,s)$ where $t_1 \sim \deriv{s}{w_1}$.
  By induction on the case $(f_2,t_1)$, there exists $t \in \FixP{e}{r}(f_2,t_1)$ where $t \sim \deriv{t_1}{w_2} \sim \deriv{s}{w_1 \conc w_2}$.
 We have that $t \in \FixP{e}{r}(f_1\conc f_2,s)$ based on the following
 reasoning.
 \bda{c}
 \FixP{e}{r}(f_1\conc f_2,s)
 \\ \supseteq
 \\ \ReachFunc{f_2}{\ReachFunc{f_1}{s}{\FixP{e}{r}}}{\FixP{e}{r}}
 \\ \supseteq
 \\ \ReachFunc{f_2}{t_1}{\FixP{e}{r}}
 \\ =
 \\ \FixP{e}{r}(f_2,t_1) \ni t_2
 \eda
 Thus, we are done for this case.

 The remaining cases are straightforward.
 \qed
\end{proof}  

\begin{lemma}
  \label{le:reach-soundness-right}
  Let $e$ be a context-free expression and $r$ be a regular expression.
  Let $\psi = \SubTerm{e}$.
  Let $n\geq 0$, $(f,s) \in \SubTerm{e} \times \Desc{r}$ and $t$ be a regular
  expression such that $t \in \FixStep{n}{e}{r}(f,s)$.
  Then, there exists word $w$ such that $\apply{\psi}{f} \reduce w$
  and $t \sim \deriv{s}{w}$.
\end{lemma}
\begin{proof}
  By induction over $n$.

  {\bf Case} $n=0$: Statement holds trivially as $\FixStep{0}{e}{r}(f,s) = \{ \}$.

  {\bf Case} $n \implies n+1$:

  We proceed by induction over the structure of $f$.

  {\bf Subcase} $\phi$: Trivial.

  {\bf Subcase} $x$:

  Consider $t \in \FixStep{n+1}{e}{r}(x,s) = \{ \simp{\deriv{s}{x}} \}$.
  Take $w = x$ and the statement is satisfied.

  {\bf Subcase} $\varepsilon$:

  Consider $t \in \FixStep{n+1}{e}{r}(x,s) = \{ \simp{s} \}$.
  Take $w = \varepsilon$ to satisfy the statement.
  
  {\bf Subcase} $\alpha$:

  Consider 
  \bda{ll}
     t \in &  \FixStep{n+1}{e}{r}(\alpha,s)
       \\ & = (\ReachFunc{e}{r}(\FixStep{n}{e}{r}))(f,s)
       \\ & = \FixStep{n}{e}{r}(\alpha,s) \cup \ReachF{\alpha}{s}{\FixStep{n}{e}{r}}
       \\ & = \FixStep{n}{e}{r}(\alpha,s) \cup \FixStep{n}{e}{r}(\mu\alpha.f,s)
   \eda
   Suppose $t \in \FixStep{n}{e}{r}(\alpha,s)$.
   By induction, there exists $w$ such that
   $\apply{\psi}{\alpha} \reduce w$ and $t \sim \deriv{s}{w}$
   and thus we can establish the statement.
   Otherwise, $t \in \FixStep{n}{e}{r}(\mu\alpha.f,s)$.
   By induction, there exists $w$ such that
   $\apply{\psi}{\mu\alpha.f} \reduce w$ and $t\sim \deriv{s}{w}$.
   By construction of $\psi$ we have that
   $\apply{\psi}{\alpha} = \apply{\psi}{\mu\alpha.f}$
   and we are done again.

   {\bf Subcase} $\mu\alpha.f$:

   Consider
   \bda{ll}
   t \in & \FixStep{n+1}{e}{r}(\mu\alpha.f,s)
   \\ & = \FixStep{n}{e}{r}(\mu\alpha.f,s) \cup \ReachF{f}{s}{\FixStep{n}{e}{r}}
   \\ & \subseteq \FixStep{n}{e}{r}(\mu\alpha.f,s) \cup \FixStep{n+1}{e}{r}(f,s)
   \eda

   Suppose $t \in \FixStep{n}{e}{r}(\mu\alpha.f,s)$.
   By induction, there exists $w$ such that $\apply{\psi}{\mu\alpha.f} \reduce w$ and $t \sim \deriv{s}{w}$. Hence, we can establish the statement.
   Otherwise, $t \in \FixStep{n+1}{e}{r}(f,s)$.
   By induction, there exists $w$ such that $\apply{\psi}{f} \reduce w$
   and $t \sim \deriv{s}{w}$.
     By Lemma~\ref{le:psi} we find that
     $\apply{[\alpha \mapsto \mu\alpha.\apply{\restrict{\psi}{\alpha}}{f}]}{\apply{\restrict{\psi}{\alpha}}{f}} = \apply{\psi}{f}$.
     By definition $\apply{\psi}{\mu\alpha.f} = \mu\alpha.\apply{\restrict{\psi}{\alpha}}{f}$.
     Hence, we can conclude that $\apply{\psi}{\mu\alpha.f} \reduce w$
     and we are done for this subcase.

   {\bf Subcase} $f_1 + f_2$:

   Consider
   \bda{ll}
   t \in & \FixStep{n+1}{e}{r}(f_1 + f_2,s)
   \\ & = (\ReachFunc{e}{r}(\FixStep{n}{e}{r}))(f_1+f_2,s)
   \\ & = \FixStep{n}{e}{r}(f_1+f_2,s)
          \cup \ReachF{f_1}{s}{\FixStep{n}{e}{r}}
          \cup \ReachF{f_2}{s}{\FixStep{n}{e}{r}}
   \eda
   Suppose $t \in \FixStep{n}{e}{r}(f_1+f_2,s)$.
   By induction, there exists $w$ such that $\apply{\psi}{f_1 + f_2}$
   and $t \sim \deriv{s}{w}$.
   Hence, the statement holds.
   Suppose $t \in \ReachF{f_1}{s}{\FixStep{n}{e}{r}}$.
   By definition $\ReachF{f_1}{s}{\FixStep{n}{e}{r}} \subseteq \FixStep{n+1}{e}{r}(f_1,s)$.
   By induction, there exists $w$ such that $\apply{\psi}{f_1} \reduce w$
   and $t \sim \deriv{s}{w}$.
   We can conclude that $\apply{\psi}{f_1 + f_2} \reduce w$ and are done.
   Otherwise, $t \in \ReachF{f_2}{s}{\FixStep{n}{e}{r}}$.
   Similar reasoning applies as in the previous case.

   {\bf Subcase} $f_1 \conc f_2$:

   Consider
   \bda{ll}
   t \in & \FixStep{n+1}{e}{r}(f_1 \conc f_2,s)
   \\ & = \FixStep{n}{e}{r}(f_1 \conc f_2,s)
          \cup \ReachF{f_2}{\ReachF{f_1}{s}{\FixStep{n}{e}{r}}}{\FixStep{n}{e}{r}}
   \eda
   Suppose $t \in \FixStep{n}{e}{r}(f_1 \conc f_2,s)$.
   By induction there exists $w$ such that $\apply{\psi}{f_1 \conc f_2}$
   and $t \sim \deriv{s}{w}$.
   Hence, the statement holds.
   Otherwise, $t \in \ReachF{f_2}{\ReachF{f_1}{s}{\FixStep{n}{e}{r}}}{\FixStep{n}{e}{r}}$.
   There exists $t_1 \in \ReachF{f_1}{s}{\FixStep{n}{e}{r}}$
   such that $t \in \ReachF{f_2}{t_1}{\FixStep{n}{e}{r}}$.
   By induction on $t_1 \in \ReachF{f_1}{s}{\FixStep{n}{e}{r}}$,
   there exists $w_1$ such that $\apply{\psi}{f_1} \reduce w_1$
   and $t_1 \sim \deriv{s}{w_1}$.
   by induction on $t \in \ReachF{f_2}{t_1}{\FixStep{n}{e}{r}}$,
   there exists $w_2$ where $\apply{\psi}{f_2} \reduce w_2$
   and $t_2 \sim \deriv{t_1}{w_2}$.
   We can conclude that $\apply{\psi}{f_1 \conc f_2} \reduce w_1 \conc w_2$
   and $t \sim \deriv{s}{w_1 \conc w_2}$ and are thus done.
   \qed
\end{proof}

\begin{lemma}
  Let $e$ be a context-free expression and $r$ be a regular expression.
  Then, we have that $\FixP{e}{r}(e,r) \sim \reach{e}{r}$
\end{lemma}
\begin{proof}
  Follows from Lemmas~\ref{le:reach-soundness-left} and \ref{le:reach-soundness-right} and the fact that $\apply{\SubTerm{e}}{e} = e$.
\end{proof}

\section{Proofs}

For some proofs we make use of the terminology and results
introduced in the above.

\subsection{Proof of Lemma \ref{le:flatten-cfg-proof}}

\begin{proof}
  By induction on the derivation $e \reduce w$.
  \qed
\end{proof}

\subsection{Proof of Lemma \ref{le:parse-trees-flattening}}

\begin{proof}
  By induction on the derivation $\turns p : e$.
  For brevity, we consider some selected cases.

  {\bf Case} $\mu\alpha.e$:

  By assumption $  \myirule{\turns p : \apply{[\alpha \mapsto \mu \alpha.e]}{e}}
  {\turns \Fold\ p : \mu \alpha.e}$.
  By induction, $\apply{[\alpha \mapsto \mu \alpha.e]}{e} \reduce \flatten{p}$.
  By definition, $\flatten{\Fold\ p} = \flatten{p}$.
  Hence, $\mu\alpha.e \reduce \flatten{\Fold\ p}$.

  {\bf Case} $e + f$:

  {\bf Subcase} $p = \Left\ p_1$:

  By induction, $e \reduce \flatten{p_1}$.
  By definition, $\flatten{\Left\ p_1} = \flatten{p_1}$.
  Hence, $e + f \reduce \flatten{p}$.

    {\bf Subcase} $p = \Right\ p_2$: Similar reasoning as above.
  \qed
\end{proof}

\subsection{Proof of Theorem \ref {th:containment-reachability}}

\begin{proof}
  By definition $e \le r$ iff
  ($\forall w
  \in\Sigma^*$, $e \reduce
  w$ implies $r \reduce w$)
  iff ($\forall w
  \in\Sigma^*$, $e \reduce
  w$ implies $\deriv rw \reduce \varepsilon$) iff
  each expression in $\reach e r$ is nullable.
  %
  %
  %
  \qed
\end{proof}

\subsection{Proof of Lemma \ref{le:sound-reach-characterization}}
\label{proof:le:sound-reach-characterization}

\begin{proof}
  We generalize the statement as follows.
  Consider $e$ and $r$ fixed.
  For $f,f' \in \SubTerm{e}$ we write $f < f'$ to denote that $f$ is a subexpression in $f'$
  where $f \not= f'$.
  Let $\psi = \Subst{e}$.
  Consider $(f,s) \in \SubTerm{e} \times \Desc{r}$.
  Let $\Gamma = \{ \reachRel{s}{\apply{\psi}{\mu\alpha.f'}}{\reach{\apply{\psi}{\mu\alpha.f'}}{s}} \mid \mu\alpha.f' \in \SubTerm{e} \wedge f < \mu\alpha.f' \}$.
  So, the environment $\Gamma$ consists of all assumptions which are in the surrounding scope of $f$.
  
  We claim that $\Gamma \turns \reachRel{s}{\apply{\psi}{f}}{\reach{\apply{\psi}{s}}{s}}$ is derivable.
  The statement follows for $e$ and $r$ from the fact that for $e$ the environment $\Gamma$ is empty and $\apply{\psi}{e} = e$.

  We verify that $\Gamma \turns \reachRel{s}{\apply{\psi}{f}}{\reach{\apply{\psi}{s}}{s}}$ is derivable
  by induction on $f$.

  {\bf Case} $\mu\alpha.f$: We observe that $\apply{\psi}{\mu\alpha.f} = \mu\alpha.\apply{\restrict{\psi}{\alpha}}{f}$.
  Hence, the desired statement
  $$\Gamma \turns \reachRel{s}{\apply{\psi}{\mu\alpha.f}}{\reach{\apply{\psi}{\mu\alpha.f}}{s}}$$
  is equal to
  $$\Gamma \turns \reachRel{s}{\mu\alpha.\apply{\restrict{\psi}{\alpha}}{f}}{\reach{\mu\alpha.\apply{\restrict{\psi}{\alpha}}{f}}{s}}.$$
  By rule inversion,
  $$\Gamma \turns \reachRel{s}{\mu\alpha.\apply{\restrict{\psi}{\alpha}}{f}}{\reach{\mu\alpha.\apply{\restrict{\psi}{\alpha}}{f}}{s}}$$
  if
  $$\Gamma \cup \{ \reachRel{s}{\mu\alpha.\apply{\restrict{\psi}{\alpha}}{f}}{\reach{\mu\alpha.\apply{\restrict{\psi}{\alpha}}{f}}{s}} \}
  \turns \reachRel{s}{\apply{[\alpha\mapsto\mu\alpha.\apply{\restrict{\psi}{\alpha}}{f}]}{\apply{\restrict{\psi}{\alpha}}{f}}}{\reach{\mu\alpha.\apply{\restrict{\psi}{\alpha}}{f}}{s}} \ (1).$$
    By assumption $\Gamma$ has the proper form for $\apply{\psi}{\mu\alpha.f}$.
    Hence, $\Gamma \cup \{ \reachRel{s}{\apply{\psi}{\mu\alpha.f}}{\reach{\apply{\psi}{\mu\alpha.f}}{s}} \}$
  has the proper form for $f$. By induction,
  $$
  \Gamma \cup \{ \reachRel{s}{\apply{\psi}{\mu\alpha.f}}{\reach{\apply{\psi}{\mu\alpha.f}}{s}} \} \turns \reachRel{s}{\apply{\psi}{f}}{\reach{\apply{\psi}{f}}{s}} \ (2).
  $$
  We observe that $\apply{\psi}{f} = \apply{[\alpha\mapsto\mu\alpha.\apply{\restrict{\psi}{\alpha}}{f}]}{\apply{\restrict{\psi}{\alpha}}{f}}$
  and $\reach{\mu\alpha.\apply{\restrict{\psi}{\alpha}}{f}}{s} = {\reach{\apply{[\alpha\mapsto\mu\alpha.\apply{\restrict{\psi}{\alpha}}{f}]}{\apply{\restrict{\psi}{\alpha}}{f}}}{s}} = \reach{\apply{\psi}{f}}{s}$.
  Hence, (1) and (2) are equal and therefore the desired statement can be derived.

  {\bf Case} $e+f$: Expressions $e$ and $f$ share the same $\Gamma$.
  By induction, $\Gamma \turns \reachRel{s}{\apply{\psi}{e}}{\reach{\apply{\psi}{e}}{s}}$ and
  $\Gamma \turns \reachRel{s}{\apply{\psi}{f}}{\reach{\apply{\psi}{f}}{s}}$.
  By rule \rlabel{Alt}, $\Gamma \turns \reachRel{s}{\apply{\psi}{e}+\apply{\psi}{f}}{\reach{\apply{\psi}{e}}{s} \cup \reach{\apply{\psi}{f}}{s}}$.
  We observe that $\apply{\psi}{e+f} = \apply{\psi}{e}+\apply{\psi}{f}$
  and $\reach{\apply{\psi}{e+f}}{s} = \reach{\apply{\psi}{e}}{s} \cup  \reach{\apply{\psi}{f}}{s}$ and are done for this case.

  {\bf Case} $e \conc f$: Expressions $e$ and $f$ share the same $\Gamma$.
  By induction, $\Gamma \turns \reachRel{s}{\apply{\psi}{e}}{\reach{\apply{\psi}{e}}{s}}$.
  Suppose $\reach{\apply{\psi}{e}}{s} = \{ \}$.
  Then, $\Gamma \turns \reachRel{s}{\apply{\psi}{e}\conc\apply{\psi}{f}}{\{\}}$.
  Under the assumption, $\reach{\apply{\psi}{e}\conc\apply{\psi}{f}}{s} = \{\}$
  and we are done.  
  Otherwise, $\reach{\apply{\psi}{e}}{s} = \{ s_1, ..., s_n \}$ for $n>0$.
  By induction, for each combination $(f,s_i)$, $\Gamma \turns \reachRel{s_i}{\apply{\psi}{f}}{\reach{\apply{\psi}{f}}{s_i}}$.
  By rule \rlabel{Seq}, $\Gamma \turns \reachRel{s}{\apply{\psi}{e}\conc\apply{\psi}{f}}{\reach{\apply{\psi}{f}}{s_1}\cup \ldots \cup \reach{\apply{\psi}{f}}{s_n}}$.
  By the fact that $\reach{\apply{\psi}{f}}{\reach{\apply{\psi}{e}}{s}} = \reach{\apply{\psi}{f}}{s_1}\cup \ldots \cup \reach{\apply{\psi}{f}}{s_n}$
  we reach the desired conclusion.

  {\bf Cases} $x$, $\varepsilon$, $\phi$: Straightforward.
  \qed
\end{proof}  

\subsection{Proof of Lemma \ref{le:complete-reach-characterization}}
\label{proof:le:complete-reach-characterization}

The proof requires a couple of technical statements.

\begin{lemma}[Strengthening]
  \label{le:weaker-env}
  Let $e$ be a context-free expression, $r$ be a regular expression, $S$ be a set
  and $\Gamma', \Gamma$ be two environments such that $\Gamma' \supseteq \Gamma$
  and $\Gamma' \turns \reachRel{r}{e}{S}$ where
  in the derivation tree the extra assumptions $\Gamma' - \Gamma$ are not used.
  Then, we also find that $\Gamma \turns \reachRel{r}{e}{S}$.
\end{lemma}
\begin{proof}
By induction on the derivation.
\end{proof}

\begin{lemma}[Weakening]
  \label{le:stronger-env}
  Let $e$ be a context-free expression, $r$ be a regular expression, $S$ be a set
  and $\Gamma', \Gamma$ be two environments such that $\Gamma' \supseteq \Gamma$
  and $\Gamma \turns \reachRel{r}{e}{S}$.
  Then, we also find that $\Gamma' \turns \reachRel{r}{e}{S}$
\end{lemma}
\begin{proof}
By induction on the derivation.
\end{proof}

\begin{lemma}[Substitution]
  \label{le:unfold-equiv}
  Let $\mu\alpha.f$ be a context-free expression, $r$ be a regular expression
  and $S$ a set such that $\turns \reachRel{r}{\mu\alpha.f}{S}$.
  Then, we find that $\turns \reachRel{r}{\apply{[\alpha \mapsto \mu\alpha.f]}{f}}{S}$.
\end{lemma}
\begin{proof}
  We generalize the statement and include some environment $\Gamma$.
  We write $D$ to denote the derivation tree for
  $\Gamma \turns \reachRel{r}{\mu\alpha.f}{S}$.
  The shape of $D$ is as follows.
  \bda{c}
  \inferrule*[left = \rlabel{Rec}]
             {
               \inferrule*[]
                          { \ldots }
                          {
                            \Gamma \cup \{ \reachRel{r}{\mu\alpha.f}{S} \} \turns \reachRel{r}{\apply{[\alpha\mapsto\mu\alpha.f]}{f}}{S}
                            }
             }
             {
               \Gamma \turns \reachRel{r}{\mu\alpha.f}{S}
             }
             
  \eda

  Suppose, in the upper derivation tree (denoted by $\ldots$),
  there are no applications of \rlabel{Hyp} for $\mu\alpha.f$.
  By Lemma~\ref{le:weaker-env}, we can immediately conclude that $\Gamma \turns \reachRel{r}{\apply{[\alpha\mapsto\mu\alpha.f]}{f}}{S}$ is derivable as well.
  Otherwise, we consider all applications of \rlabel{Hyp} for $\mu\alpha.f$.
  In the below, we show only one such application.

  \bda{c}
  \inferrule*[left = \rlabel{Rec}]
             {
               \inferrule*[]
                          { \inferrule*[left = \rlabel{Hyp}]
                            { \Gamma' \cup \{ \reachRel{r}{\mu\alpha.f}{S} \} \turns \reachRel{r}{\mu\alpha.f}{S}
                            }
                            { \ldots }
                            }
                          {
                            \Gamma \cup \{ \reachRel{r}{\mu\alpha.f}{S} \} \turns \reachRel{r}{\apply{[\alpha\mapsto\mu\alpha.f]}{f}}{S}
                            }
             }
             {
               \Gamma \turns \reachRel{r}{\mu\alpha.f}{S}
             }
             
             \eda
             where by construction $\Gamma' \supseteq \Gamma$.

             Each such \rlabel{Hyp} rule application can be replaced
             by the derivation tree $D$ where we make use of $\Gamma' \cup \{ \reachRel{r}{\mu\alpha.f}{S} \}$ instead
             of $\Gamma$ (justified by Lemma~\ref{le:stronger-env}).
             In fact, we can argue that the extra assumption $\reachRel{r}{\mu\alpha.f}{S}$ is no longer required due to the elimination of rule \rlabel{Hyp}.
             Hence, we can argue that $\Gamma \turns \reachRel{r}{\apply{[\alpha\mapsto\mu\alpha.f]}{f}}{S}$ is derivable.
\qed
\end{proof}

We write $\unfold{k}{e}$ to denote that all recursive constructs in $e$
have been unfolded at least $k$-times.

\begin{lemma}
  \label{le:unfold-approx}
  Let $e$ be a context-free expression, $r$ be a regular expression.
  Then, for any $n \geq 0$ there exists a $k$ such that
  $S \supseteq \FixStep{n}{e}{r}(e,r)$ where
  $\turns \reachRel{r}{\unfold{k}{e}}{S}$.
\end{lemma}
\begin{proof}

  We define

  \bda{lcl}
  \ReachFUnfold{\phi}{r} & = & \{ \}
  \\
  \ReachFUnfold{\varepsilon}{r} & = & \{ \simp{r} \}
  \\
  \ReachFUnfold{x}{r} & = & \{ \simp{\deriv{r}{x}} \}
  \\
  \ReachFUnfold{e + f}{r} & = & \ReachFUnfold{e}{r} \cup \ReachFUnfold{f}{r}
  \\
  \ReachFUnfold{e \conc f}{r} & = & \bigcup_{s \in \ReachFUnfold{e}{r}} \ReachFUnfold{f}{s}
  \\
  \ReachFUnfold{\alpha}{r} & = & \{ \} 
  \\
  \ReachFUnfold{\mu\alpha.e}{r} & = & \{ \}
  \eda

  S1:
  For a fixed $e$ and $r$, for any $(f,s) \in \SubTerm{e} \times \Desc{r}$
  and $n \geq 0$, there exists $k$ such that
  $\ReachFUnfold{\unfold{k}{(\apply{\psi}{f})}}{s} \supseteq \FixStep{n}{e}{r}(f,s)$.
  Like the proof of Lemma \ref{le:reach-soundness-right},
  we verify the statement by applying
  induction over $n$ and observing the structure of $f$.

  {\bf Case} $n=0$: Straightforward.
  
  {\bf Case} $n \implies n+1$:

  We proceed by induction over the structure of $f$.
  
  {\bf Subcase} $\mu\alpha.f$.
  We have that $\FixStep{n+1}{e}{r}(\mu\alpha.f,s) \subseteq \FixStep{n}{e}{r}(\mu\alpha.f,s) \cup \FixStep{n+1}{e}{r}(f,s)$.
  By induction on $n$, $\ReachFUnfold{\unfold{k_1}{(\apply{\psi}{\mu\alpha.f})}}{s} \supseteq  \FixStep{n}{e}{r}(\mu\alpha.f,s)$ for some $k_1$.
  By induction on $f$, $\ReachFUnfold{\unfold{k_2}{(\apply{\psi}{f})}}{s} \supseteq \FixStep{n+1}{e}{r}(f,s)$ for some $k_2$.
  Recall that $\apply{[\alpha \mapsto \mu\alpha.\apply{\restrict{\psi}{\alpha}}{f}]}{\apply{\restrict{\psi}{\alpha}}{f}} = \apply{\psi}{f}$
  and $\apply{\psi}{\mu\alpha.f} = \mu\alpha.\apply{\restrict{\psi}{\alpha}}{f}$.
  Hence, $\unfold{1}{(\apply{\psi}{\mu\alpha.f})} = \apply{\psi}{f}$.
  Function $\ReachFUnfold{}{}$ is a monotone function respect to unfoldings.
  We set $k = k_1 + k_2$.
  Then, $\ReachFUnfold{\unfold{k}{(\apply{\psi}{\mu\alpha.f})}}{s} \supseteq  \FixStep{n+1}{e}{r}(\mu\alpha.f,s)$ and we are done for this case.

  S2:
  For $\turns \reachRel{r}{\unfold{k}{e}}{S}$, we have that $S \supseteq \ReachFUnfold{\unfold{k}{e}}{r}$.
  By induction on $k$ and observing the structure of $e$.
  \ms{todo, should work out similarly as in the case above.}

  Desired statement follows from S1 and S2.
  \qed
\end{proof}

We are in the position
to proof Lemma~\ref{le:complete-reach-characterization}.
We recall the statement of this proposition:
  Let $e$ be a context-free expression, $r$ be a regular expression
  and $S$ be a set of expressions such that $\turns \reachRel{r}{e}{S}$.
  Then, we find that $S \supseteq \reach{e}{r}$.
\begin{proof}
  Assume the contrary. Then, there exists $s \in S$ and
  $s \not \in \FixStep{n}{e}{r}(e,r)$ for some $n \geq 0$.
  By Lemma \ref{le:unfold-equiv}, we find that
  $\turns \reachRel{r}{e^k}{S}$ for any $k$.
  By Lemma \ref{le:unfold-approx}, $S \supseteq \FixStep{n}{e}{r}(e,r)$
  which contradicts the assumption.
  \qed
\end{proof}

Based on the above, we obtain a greatest fixpoint method to
compute $\reach{e}{r}$.
We consider $e$ and $r$ fixed.
For each combination $(f,s) \in \SubTerm{e} \times \Desc{r}$,
we set the respective $S$ to $\Desc{r}$.
In each greatest fixpoint step, we pick a combination where we remove
one of the elements in $S$. Check if
$\turns \reachRel{f}{s}{S}$~\footnote{Need to include the environment,
  apply $\psi$, as we already start if with some environment,
  can only apply \rlabel{Hyp} after one application of \rlabel{Rec} ...}
is still derivable.
If yes, continue the process of eliminating elements.

\subsection{Proof of Lemma \ref{le:upCoerce}}

\ms{FIXME: semantic well-behaved ....}

\begin{definition}[Well-Behaved Upcast]
  Let $e$ be a context-free expression, $r$ be a regular expression,
  and $c$ be a coercion of type $\upTy{e}{r}$, where we write
  $\upTy{e}{r}$ for the type $(e, \plus{\reach{e}{r}}) \rightarrow r$.

  We say $c$ is a \emph{well-behaved upcast} iff
  for any $\turns p : e$ and $\turnsreg t : \plus{\reach{e}{r}}$ we
  find that $\turnsreg \fapp{c}{p,t} : r$.

  We further define environments $\Delta$ by
   \bda{lcl}
    \Delta & ::= & \{ \} \mid \{ v : \upTy{e}{r} \} \mid \Delta \cup \Delta
  \eda

  We say that $\Delta$ is a \emph{well-behaved upcast} environment iff
  each $(v : \upTy{e}{r}) \in \Delta$ is a well-behaved upcast coercion.
\end{definition}


\begin{lemma}[Soundness]
    \label{le:upCoerce}
  Let $\Delta$ be a well-behaved upcast environment.
  Let $e$ be a context-free expression and $r$ be a regular expression
  such that $\upCoerce{\Delta}{c}{e}{r}$ for some coercion $c$.
  Let $p$ and $t$ be parse trees such that $\turns p : e$
  and $\turnsreg t : \plus{\reach{e}{r}}$ where
  $\flatten{t} \in L(\deriv{r}{\flatten{p}})$.
  Then, we find that $\Delta \turns \fapp{c}{(p,t)} : r$.
\end{lemma}

\begin{proof}
  By induction on the derivation to construct coercions.
  For brevity, we sometimes omit $\Gamma$ in case it is not relevant.

  \textbf{Case }$\varepsilon$:
  By assumption $\turns p : \varepsilon$ and $\turns t : \plus{\reach{\varepsilon}{r}}$.

  Thus, $p = \Eps$ and $t : \simp{r}$.

  Inversion yields a regular coercion $c_1 : \simp{r} \to r$.

  Hence $(\lambda (\Eps, q). \fapp{c_1}{t}) (\varepsilon, t) = \fapp{c_1}{t}$ with $\turns \fapp{c_1}{t} : r$.

  \textbf{Case }$x$:
  By assumption $\turns p : x$ and $\turns t : \plus{\reach{x}{r}}$.

  Thus $p = \Sym\ x$ and $\turns t : \plus{\reach{x}{r}}$ where $\plus{\reach{x}{r}} = \simp{\deriv{r}{x}}$.

  Inversion yields a regular coercion $c : (x, \simp{\deriv{r}{x}}) \to r$.
  
  Hence $\turns \fapp{c}{(p, t)} : r$

  \textbf{Case} $e\conc f$:
  By assumption $\turns p : e \conc f$ and $\turns t: \plus{\reach{e \conc f}{r}}$ and
  $\flatten{t} \in L(\deriv{r}{\flatten{p}})$.

  Inversion for $p$ yields $p = \Seq{p_1}{p_2}$ such that $\turns p_1 : e$ and $\turns p_2 : f$.
  It holds that $\flatten{p} = \flatten{p_1} \conc \flatten{p_2}$.

  Further,
  \begin{align*}
    \reach{e \conc f}{r} & = \bigcup \{ \reach{f}{s} \mid s\in \reach{e}{r} \}
    \\ & = \reach{f}{\plus{\reach{e}{r}}}
  \end{align*}
  so that $\turns t : \plus{\reach{f}{\plus{\reach{e}{r}}}}$.

  Now
  \begin{align*}
    \flatten{t} \in L(\deriv{r}{\flatten{p}})
    &=  L(\deriv{r}{\flatten{p_1} \conc \flatten{p_2}}) \\
    &\subseteq L (\deriv{\plus{\reach{e}{r}}}{\flatten{p_2}})
  \end{align*}

  Inversion on the coercion derivation yields
  \begin{gather*}
    \upCoerce{\Gamma}{c_1}{e}{r} \\
    \upCoerce{\Gamma}{c_2}{f}{\plus{\reach{e}{r}}}
  \end{gather*}

  Induction on the derivation of $c_2$ yields $\turns \fapp{c_2}{(p_2, t)} : \plus{\reach{e}{r}}$.

  Induction on the derivation of $c_2$ using $p_1$ for $p$ and $\fapp{c_2}{(p_2, t)}$ for $t$ yields
  $\turns \fapp{c_1}{(p_1, \fapp{c_2}{(p_2, t)})} : r$ as desired.

  {\bf Case} $e+f$: By assumption $\turns p : e + f$ and $\turns t: \plus{\reach{e + f}{r}}$ and
  $\flatten{t} \in L(\deriv{r}{\flatten{p}})$. 
  We distinguish among the following subcases.

  {\bf Subcase} $p = \Left\ p_1$:
  At this point, we have $\turns p_1 : e$ by inversion of the assumption. 
  We conclude that $\flatten{p}=\flatten{p_1} \in L(e)$.
  By Lemma~\ref{le:parse-trees-flattening}, $e \reduce \flatten{p_1}$ and
  therefore we find that $\deriv{r}{\flatten{p_1}}$ is similar to an element of $\reach{e}{r}$.  
  Because $\flatten{t} \in L(\deriv{r}{\flatten{p}})$ we conclude that $\flatten{t} \in L(\plus{\reach{e}{r}})$.
  By Lemma~\ref{le:regular-coercions}, it must be that
  $\fapp{b_1}{t}  = \Just\ t_1$ for some $t_1$ where $\turns t_1 : \plus{\reach{e}{r}}$.
  By induction we find that $\turns \fapp{c_1}{(p_1,t_1)} : r$.
  By combining the above results, we conclude that $\fapp{c}{(p,t)} : r$, too.

  {\bf Subcase} $p = \Right\ p_2$: Analogously.

  {\bf Case} $\mu\alpha. e$:

  By assumption $\turns p : \mu\alpha.e$ and $\turns t : \plus{\reach{\mu\alpha.e}{r}}$ and
  $\flatten{t} \in L(\deriv{r}{\flatten{p}})$.

  By inversion of the assumption, $p = \Fold\ p'$, $\turns p' :
  \apply{[\alpha \mapsto \mu \alpha.e]}{e}$, and $\flatten{p} = \flatten{p'}$.

  Further, as $\reach{\mu\alpha.e}{r} = \reach{\apply{[\alpha \mapsto \mu
      \alpha.e]}{e}}{r}$  (by Lemma~\ref{lemma:unfold-cfe}), we obtain
  $\turns t : \plus{\reach{\apply{[\alpha \mapsto \mu 
        \alpha.e]}{e}}{r}}$. 
  
  \textbf{Subcase:} If $(\mu\alpha.e , r)$ is not in $\Gamma$, then inversion yields some $c'$ such that
  \begin{gather*}
    \upCoerce{\Gamma \cup \{ v : (\mu\alpha.e, r) \}}{c'}{\apply{[\alpha\mapsto \mu\alpha.e]}e}{r}
  \end{gather*}

  Hence, induction is applicable observing $\turns p' : \apply{[\alpha
    \mapsto \mu \alpha.e]}{e}$, $\turns t : \plus{\reach{\apply{[\alpha \mapsto \mu 
        \alpha.e]}{e}}{r}}$, and that the flattening assumption holds.
  Thus, $\turns \fapp{c'}{(p', t)} : r$ and we need to show that

  \begin{align*}
    & \turns (\REC\ v. \lambda (\Fold\ p', t). \fapp{c'}{(p',t)}) (\Fold\ p', t) : r
    \\ \Leftrightarrow& \qquad \text{By subject reduction}
    \\ & \turns (\lambda (\Fold\ p', t). \fapp{c'}{(p',t)}[v
    \mapsto (\REC\ v. \lambda (\Fold\ p', t). \fapp{c'}{(p',t)})]) (\Fold\
    p', t) : r
    \\ \Leftrightarrow& \qquad \text{By subject reduction}
    \\ & \turns \fapp{c'}{(p',t)}[v \mapsto (\REC\ v. \lambda (\Fold\ p', t). \fapp{c'}{(p',t)})] : r
    \\ \Leftrightarrow& \qquad \text{Substitution lemma backwards}
    \\ & v : (\mu\alpha.e, \plus{\reach{\mu\alpha.e}{r}}) \to r
    \turns \fapp{c'}{(p',t)} : r
  \end{align*}
  The last statement holds and thus awe are done.

  \textbf{Subcase:} If $v : (\mu\alpha.e, r) \in \Gamma$, then the statement
   holds immediately. 
  \qed
\end{proof}  

\subsection{Proof of Lemma \ref{le:cfe-nullability}}

\begin{proof}
  The direction from left to right follows immediately.

  Suppose $\varepsilon \in L(e)$ which implies $e \reduce \varepsilon$.
  Consider the case
  $  \myirule{\apply{[\alpha \mapsto \mu \alpha.e]}{e} \reduce \varepsilon}
  {\mu \alpha.e \reduce \varepsilon}$.
  We argue that if $\apply{[\alpha \mapsto \mu \alpha.e]}{e} \reduce \varepsilon$
  then $e\reduce \varepsilon$.
  This can be verified by the number of unfolding steps
  applied on $\mu\alpha.e$.

  Consider $\apply{[\alpha \mapsto \mu \alpha.e]}{e} \reduce \varepsilon$
  where no further unfolding steps are executed on $\mu\alpha.e$.
  Then, by induction on the derivation $\apply{[\alpha \mapsto \mu \alpha.e]}{e} \reduce \varepsilon$ we obtain a derivation $e \reduce \varepsilon$.
  For clarity, we write $f \reduceNoFold{\alpha} \varepsilon$ to denote a derivation
  in which no unfolding step takes place for $\mu\alpha.e$.

  Consider the induction step.
  By induction, if $\apply{[\alpha \mapsto \mu \alpha.e]}{e} \reduce \varepsilon$
  then $\apply{[\alpha \mapsto \mu \alpha.e]}{e} \reduceNoFold{\alpha} \varepsilon$.
  By induction on the derivation $\reduceNoFold{}$ we can argue that
  we obtain $e \reduceNoFold{\alpha} \varepsilon$.
  Hence, if $\apply{[\alpha \mapsto \mu \alpha.e]}{e} \reduce \varepsilon$
  then $e\reduceNoFold{\alpha} \varepsilon$.

  The above reasoning considers the elimination of the unfolding step
  for a specific $\mu\alpha.e$.
  By induction, we can argue that any unfolding step can eliminated.
  We write $\reduceNoFold{}$ to denote the derivation where no unfolding steps occur.
  Hence, if $e \reduce \varepsilon$ then $e \reduceNoFold{} \varepsilon$.
  By induction on $\reduceNoFold{}$ we can argue that $\nullable{e}$ holds.
  \qed
\end{proof}

\subsection{Proof of Lemma \ref{le:mkEmpty}}

\begin{proof}
  Based on the observations in the proof of Lemma~\ref{le:cfe-nullability},
  from $\varepsilon \in L(e)$ we can conclude $e \reduceNoFold{} \varepsilon$.
  Then, by induction on $\reduceNoFold{}$ we can derive the desired statements.
  \qed
\end{proof}

\subsection{Proof of Theorem \ref{th:upcast-coercions}}

\begin{proof}
  Follows immediately from Lemmas \ref{le:upCoerce} and \ref{le:mkEmpty}.
  Note that all elements in $\reach{e}{r}$ are nullable.
  Hence, $\flatten{\mkEmpty{\plus{\reach{e}{r}}}} \in \deriv{r}{\flatten{p}}$.
\end{proof}

\subsection{Proof of Lemma \ref{le:downCoerce}}

\ms{polish}

\begin{definition}[Well-Behaved Downcast]
  Let $e$ be a context-free expression and $r$ be a regular expression.
  We write $\dnTy{e}{r}$ to denote the type
$r \rightarrow \Maybe(e, \plus{\reach{e}{r}})$.
Let $c$ be a coercion of type $\dnTy{e}{r}$.
We say $c$ is a \emph{well-behaved downcast} coercion iff
(1) for any $\turnsreg t : r$  we
find that $\turns \fapp{c}{t} : \Maybe(e, \plus{\reach{e}{r}})$.
  Moreover, (2) if $\fapp{c}{t} = \Just (p, t')$, then $\flatten{t} =
  \flatten{p} \conc \flatten{t'}$.
  (3) If $\fapp{c}{t} = \Nothing$, then there exist no $\turns p:e$ and
  $\turnsreg t' : \plus{\reach{e}{r}}$ such that $\flatten{t} =
  \flatten{p} \conc \flatten{t'}$.

We say that $\Delta$ is a \emph{well-behaved downcast} environment iff
each $(v : \dnTy{e}{r}) \in \Delta$ is a well-behaved downcast coercion.
\end{definition}

In case (2) holds, we can conclude that
$\flatten{u} \in L(\deriv{r}{\flatten{p}})$
due to $\turns t : r$ and Lemma~\ref{le:parse-trees-flattening}
and the fact that derivatives denote left quotients.

\begin{lemma}
  \label{le:downCoerce}
  Let $\Delta$ be a well-behaved downcast environment, $e$ be a context-free expression,  $r$ a regular expression, and $c$ a coercion
  such that $\downCoerce{\Delta}{c}{e}{r}$.
  (1) For each regular parse tree $\turnsreg t : r$
  we find that $\Delta \turns \fapp{c}{t} : \Maybe(e, \plus{\reach{e}{r}})$.

  Moreover, (2) if $\fapp{c}{t} = \Just (p, t')$, then $\flatten{t} =
  \flatten{p} \conc \flatten{t'}$ and 
  $\flatten{t'} \in \deriv{r}{\flatten{p}}$.
  (3) If $\fapp{c}{t} = \Nothing$, then there exist no $\turns p:e$ and
  $\turnsreg u : \plus{\reach{e}{r}}$ such that $\flatten{t} =
  \flatten{p} \conc \flatten{u}$ and $\flatten{u} \in \deriv{r}{\flatten{p}}$.
\end{lemma}

\begin{proof}
  By induction on the downcast coercion derivation.
  
  \textbf{Case }$\Eps$. (1) follows easily.
  For (2), we find that
  $\flatten{\Eps} \conc \flatten{b\ (t)} = \varepsilon \conc \flatten{t}
  = \flatten{t}$. Case (3) never arises here.

  \textbf{Case }$\Sym\ x$. (1) is again straightforward.
  Case (2), suppose $\fapp{c}{t} = \Just (\Sym\ x, t')$, then 
  $\flatten{\Sym\ x} \conc \flatten{t'} = x \conc \flatten{t'} = \flatten{t}$.
  Case (3), suppose $\fapp{c}{t} = \Nothing$.
  Then, $\fapp{b}{t} = \Nothing$ which implies $\deriv{r}{x}$ denotes the empty language
  and therefore no $p$, $t'$ with the desired property can exist.



  \textbf{Case }$e+f$.
  By induction we obtain $c_1$ such that $\Gamma \turns \fapp{c_1}{t} : \Maybe(e, \plus{\reach{e}{r}})$;
  if $\fapp{c_1}{t} = \Just (p_1, t_1)$, then
  $flatten{t} = \flatten{p_1} \conc \flatten{t_1}$.
  In this case, $\fapp{c}{t} = \Just (\Left \ p_1, \fapp{b_1}{t_1})$.
  By  property of regular coercions
  (see Lemma~\ref{le:regular-coercions}),
  we find that $\flatten{t} = \flatten{p_1} \conc \flatten{\fapp{b_1}{t_1}} = \flatten{p_1} \conc \flatten{t_1}$.
  If $\fapp{c_1}{t} = \Nothing$, then no such $p_1$ and $t_1$ exist
  and therefore also no $t_1'$ of the form $\turns t_1' : \plus{\reach{e+f}{r}}$
  exists.

  Similar reasoning applies to $c_2$.
  On the other hand, if $\fapp{c}{t} = \Just (p_1,t_1')$ then one of the respective
  cases of either $c_1$ or $c_2$ applies.
  Similar reasoning applies in case of $\fapp{c}{t} = \Nothing$
  as any $\turns p : e+f$ has either form
  $\Left\ p_1$ or $\Right\ p_2$, where neither suitable $p_1$ nor $p_2$ exist,
  no such $p$ can
  exist either.

  \textbf{Case }$e \conc f$.
  By induction, we obtain $c_1, c_2$ such that
  $\Gamma \turns \fapp{c_1}{t} : \Maybe(e, \plus{\reach{e}{r}})$ and
  $\Gamma \turns \fapp{c_2}{t_1} :  \Maybe(f,  \plus{\reach{f}{\plus{\reach{e}{r}}}})$.
  Hence, $\Gamma \turns \fapp{c}{t} : \Maybe (e\conc f, \plus{\reach{e\conc f}{r}})$.

  Suppose $\fapp{c}{t} = \Just (p,t_2)$.
  Then, $\fapp{c_1}{t} = \Just (p_1,t_1)$
  and $\fapp{c_2}{t_1} = \Just (p_2,t_2)$.
  By induction, $\flatten{t} = \flatten{p_1} \conc \flatten{t_1}$
  and $\flatten{t_1} = \flatten{p_2} \conc \flatten{t_2}$.
  Hence, $\flatten{t} = \flatten{\Seq{p_1}{p_2}} \conc \flatten{t_2}$.

  Also, $\flatten{t_1} \in \deriv{r}{\flatten{p_1}}$
  and $\flatten{t_2} \in \deriv{\plus{\reach{e}{r}}}{\flatten{p_2}}$,  by induction,
  but as  $\flatten{t_1} = \flatten{p_2} \conc \flatten{t_2}$, we know
  that $\flatten{t_2} \in
  \deriv{\deriv{r}{\flatten{p_1}}}{\flatten{p_2}} =  \deriv{r}{\flatten{\Seq{p_1}{p_2}}}$.

  Suppose $\fapp{c}{t} = \Nothing$.
  Then, either $\fapp{c_1}{t} = \Nothing$
  or if $\fapp{c_1}{t} = \Just (p_1,t_1)$ then $\fapp{c_2}{t_1} =
  \Nothing$.

  If  $\fapp{c_1}{t} = \Nothing$, then there exist no $\turns p_1:e$
  and $\turns t_1 : \plus{\reach{e}{r}}$ such that $\flatten{t} =
  \flatten{p_1}\conc\flatten{t_1}$.
  If there was some suitable $p = \Seq{p_1}{p_2}$ with $\turns
  p:e\conc f$, then we can construct suitable $t_1$ for
  $p_1$. Contradiction.

  If $\fapp{c_1}{t} = \Just (p_1,t_1)$ and $\fapp{c_2}{t_1} =
  \Nothing$, then we can derive a similar contradiction for $p_2$ and $t_2$.

  \textbf{Case} $\mu\alpha.e$.
  By induction, $\Gamma\cup \{ v : (\mu\alpha.e,r) \} \turns \fapp{c'}{t} : \Maybe(\apply{[\alpha\mapsto\mu\alpha.e]}{e}, \plus{\reach{\apply{[\alpha\mapsto\mu\alpha.e]}{e}}{r}})$.
  By subjection reduction, $\Gamma \turns \fapp{c}{t} : \Maybe(\mu\alpha.e,\plus{\reach{\mu\alpha.e}{r}})$.

  Suppose $\fapp{c}{t} = \Just (p,t')$.
  We find that $p = \Fold\ p'$.
  We have that $\plus{\reach{\apply{[\alpha\mapsto\mu\alpha.e]}{e}}{r}} = \plus{\reach{\mu\alpha.e}{r}}$.
  Hence, $\turns t' : \plus{\reach{\mu\alpha.e}{r}}$ implies $\turns t' : \plus{\reach{\apply{[\alpha\mapsto\mu\alpha.e]}{e}}{r}}$ and vice versa. 
  By induction, $\flatten{t} = \flatten{p'} \conc \flatten{t'}$ and this establishes (2).
  Suppose $\fapp{c}{t} = \Nothing$.
  Then, $\fapp{c'}{t} = \Nothing$.
  By induction, no suitable $p'$ and $t'$ exists
  where $\turns p' : \apply{[\alpha\mapsto\mu\alpha.e]}{e}$
  and $\turns t' : \plus{\reach{\apply{[\alpha\mapsto\mu\alpha.e]}{e}}{r}}$.
  Hence, there can be no suitable $p$ and $t'$ either where
  $\turns p : \mu\alpha.e$ and $\turns t' : \plus{\reach{\mu\alpha.e}{r}}$
  and we are done.
  \qed
\end{proof}

\end{document}